\documentclass[a4paper,runningheads,USenglish]{llncs}

\newif\ifshowproofs
\showproofstrue

\ifshowproofs{
  
}
\else  
  \usepackage{fullpage}
  \usepackage{environ}
  \NewEnviron{hide}{}

\fi

\usepackage{fullpage}

\usepackage[export]{adjustbox}
\usepackage{subfig}
\usepackage{microtype}
\usepackage{amsmath}
\usepackage{amssymb}
\usepackage{thm-restate}
\usepackage{wrapfig}

\newcommand{\smallbox}[1]{\ensuremath{Small({#1})}}
\newcommand{\mediumbox}[1]{\ensuremath{Med}({#1})}
\newcommand{\largebox}[1]{\ensuremath{Large}({#1})}
\usepackage{dsfont}
\usepackage{xspace}
\usepackage{color}

\definecolor{blue}{rgb}{0.274,0.392,0.666}
\definecolor{darkblue}{rgb}{0.063,0.306,0.545}
\definecolor{red}{rgb}{1,0.3,0.3}
\definecolor{greennn}{rgb}{0,0.588,0.509}

\newcommand{\blue}[1]{{{\textcolor{darkblue}{#1}\xspace}}}

\usepackage{hyperref}
\hypersetup{colorlinks,linkcolor={darkblue},citecolor={darkblue},urlcolor={darkblue}} 
\usepackage{paralist}
\usepackage[color=yellow]{todonotes}
\usepackage{rotating}
\usepackage{algorithm}
\usepackage[noend]{algpseudocode}
\usepackage{multirow}
\usepackage{dsfont}
\usepackage{rotating}
\usepackage{stmaryrd}
\usepackage{rotating}
\usepackage{lineno}
\usepackage{framed}

\usepackage{ntheorem}
\newtheorem{observation}{Observation}
\renewtheorem{claim}{Claim}[theorem]

\renewcommand{\paragraph}[1]{\smallskip\noindent\textbf{#1}\xspace}

\renewcommand{\emph}[1]{\blue{\em #1}\xspace}

\usepackage[nameinlink,capitalise]{cleveref}
\Crefname{observation}{Observation}{Observations}
\Crefname{algorithm}{Algorithm}{Algorithms}
\Crefname{section}{Section}{Sections}
\Crefname{lemma}{Lemma}{Lemmata}
\Crefname{claim}{Claim}{Claims}
\Crefname{figure}{Fig.}{Figs.}
\Crefname{figure}{Fig.}{Figs.}
\Crefname{property}{Property}{Properties}
\Crefname{enumi}{Condition}{Conditions.}

\renewcommand{\paragraph}[1]{\smallskip\noindent\textbf{#1}\xspace}

\newcommand{\remove}[1]{}

\newif\ifshort
\shorttrue
\shortfalse

\graphicspath{{figures/}}

\begin{document}

\authorrunning{F. Barrera-Cruz et al.}

\title{
How to Morph a Tree on a Small Grid}
\author{
Fidel {Barrera-Cruz}$^1$,
Manuel Borrazzo$^2$,
Giordano {Da Lozzo}$^2$,
Giuseppe \mbox{Di Battista}$^2$, 
Fabrizio Frati$^2$,
Maurizio Patrignani$^2$, and 
{Vincenzo Roselli}$^2$ \thanks{This research was supported in part by MIUR Project ``MODE'' under PRIN 20157EFM5C, by MIUR Project ``AHeAD'' under PRIN 20174LF3T8, by H2020-MSCA-RISE project 734922 -- ``CONNECT'', and by MIUR-DAAD JMP N$^\circ$ 34120.\xspace}}
\institute{
$^1$Sunnyvale, CA, USA\\
\href{mailto:fidel.barrera@gmail.com}{fidel.barrera@gmail.com}\\
$^2$Roma Tre University, Rome, Italy\\
\href{mailto:manuel.borrazzo@uniroma3.it,giordano.dalozzo@uniroma3.it,giuseppe.dibattista@uniroma3.it,fabrizio.frati@uniroma3.it,maurizio.patrignani@uniroma3.it,vincenzo.roselli@uniroma3.it}{name.lastname@uniroma3.it}
}
{\def\addcontentsline#1#2#3{}\maketitle}

\begin{abstract}
In this paper we study planar morphs between straight-line planar grid drawings of trees. A morph consists of a sequence of morphing steps, where in a morphing step vertices move along straight-line trajectories at constant speed. We show how to construct planar morphs that simultaneously achieve a reduced number of morphing steps and a polynomially-bounded resolution. We assume that both the initial and final drawings lie on the grid and we ensure that each morphing step produces a grid drawing; further, we consider both upward drawings of rooted trees and drawings of arbitrary trees.
\end{abstract}

\section{Introduction}\label{se:introduction}

The problem of morphing combinatorial structures is a consolidated research topic with important applications in several areas of Computer Science such as Computational Geometry, Computer Graphics, Modeling, and Animation. The structures of interest typically are drawings of graphs; a \emph{morph} between two drawings $\Gamma_0$ and $\Gamma_1$ of the same graph $G$ is defined as a continuously changing family of drawings $\{\Gamma_t\}$ of $G$ indexed by time $t \in [0,1]$, such that the drawing at time $t=0$ is $\Gamma_0$ and the drawing at time $t=1$ is $\Gamma_1$. 
A morph is usually required to preserve a certain drawing standard and pursues certain qualities.


The \emph{drawing standard} is the set of the geometric properties that are maintained at any time during the morph. For example, if both $\Gamma_0$ and $\Gamma_1$ are planar drawings, then the drawing standard might require that all the drawings of the morph are planar.
Other properties that might required to be preserved are the convexity of the faces, or the fact that the edges are 
straight-line segments, or polylines, etc. 

Regarding the \emph{qualities} of the morph, the research up to now mainly focused on limiting the number of \emph{morphing steps}, where in a morphing step vertices move along straight-line trajectories at constant speed. A morph $\mathcal{M}$ can then be described as a sequence of drawings $\mathcal M=\langle \Gamma_0=\Delta_0,\Delta_1,\dots,\Delta_{k}=\Gamma_1\rangle$ where the morph $\langle \Delta_{i-1}, \Delta_i \rangle$, for $i=1, \dots, k$, is a morphing step. Following the pioneeristic works of Cairns and Thomassen~\cite{c-dprc-44,t-dpg-83}, most of the literature focused on the straight-line planar drawing standard. A sequence of recent results in~\cite{DBLP:journals/siamcomp/AlamdariABCLBFH17,DBLP:conf/soda/AlamdariACBFLPRSW13,DBLP:conf/icalp/AngeliniLBFPR14,DBLP:conf/compgeom/AngeliniLFLPR15,DBLP:conf/gd/AngeliniFPR13} proved that a linear number of morphing steps suffices, and is sometimes necessary, to construct a morph between any two straight-line planar drawings of a graph. 

Although the results mentioned in the previous paragraph establish strong theoretical foundations for the topic of morphing graph drawings, they produce morphs that are not appealing from a visualization perspective. Namely, 
such algorithms produce drawings that have poor \emph{resolution}, i.e., they may have an exponential ratio of the distances between the farthest and closest pairs of geometric objects (points representing nodes or segments representing edges), even if the same ratio is polynomially bounded in the initial and final drawings. Indeed, most of the above cited papers mention the problem of constructing morphs with bounded resolution as the main challenge in this research area.

The only paper we are aware of where the resolution problem has been successfully addressed is the one by Barrera-Cruz et al.~\cite{DBLP:conf/gd/Barrera-CruzHL14}, who showed how to construct a morph with polynomially-bounded resolution between two \emph{Schnyder drawings} $\Gamma_0$ and $\Gamma_1$ of the same planar triangulation. 
The model they use in order to ensure a bound on the resolution requires that $\Gamma_0=\Delta_0,\Delta_1,\dots,\Delta_{k}=\Gamma_1$ are \emph{grid drawings}, i.e., vertices have integer coordinates, and 
the resolution is measured by comparing the area of $\Gamma_0$ and $\Gamma_1$ with the area of the $\Delta_i$'s.
We remark that morphs between planar orthogonal drawings of maximum-degree-$4$ planar graphs, like those in~\cite{DBLP:journals/talg/BiedlLPS13,DBLP:conf/compgeom/GoethemV18}, inherently have polynomial resolution.

In this paper we show how to construct morphs of tree drawings that simultaneously achieve a reduced number of morphing steps and a polynomially-bounded resolution. Adopting the setting of~\cite{DBLP:conf/gd/Barrera-CruzHL14}, we assume that $\Gamma_0$ and $\Gamma_1$ are grid drawings and we ensure that each morphing step produces a grid drawing.

We present three algorithms. The first two algorithms construct morphs between any two strictly-upward straight-line planar grid drawings $\Gamma_0$ and $\Gamma_1$ of $n$-node rooted trees; \emph{strictly-upward} drawings are such that each node lies above its children. Both algorithms construct morphs in which each intermediate grid drawing has linear width and height, where the input size is measured by $n$ and by the width and the height of $\Gamma_0$ and $\Gamma_1$. The first algorithm employs $\Theta(n)$ morphing steps. The second algorithm employs $\Theta(1)$ morphing steps, however it only applies to binary trees. The third algorithm allows us to achieve our main result, namely that for any two straight-line planar grid drawings $\Gamma_0$ and $\Gamma_1$ of an $n$-node tree, there is a planar morph with $\Theta(n)$ morphing steps between $\Gamma_0$ and $\Gamma_1$ such that each intermediate grid drawing has polynomial area, where the input size is again measured by $n$ and by the width and the height of $\Gamma_0$ and $\Gamma_1$.


The first algorithm uses recursion; namely, it eliminates a leaf in the tree, it recursively morphs the drawings of the remaining tree and it then reintroduces the removed leaf in suitable positions during the morph. The second algorithm morphs the given drawings by independently changing their $x$- and $y$-coordinates; this technique is reminiscent of a recent paper by Da Lozzo et al.~\cite{DBLP:conf/gd/LozzoBFPR18}. Finally, the third algorithm scales the given drawings up in order to make room for a bottom-up modification of each drawing into a ``canonical'' drawing of the tree.

We remark that, although tree drawing algorithms are well investigated in Graph Drawing, morphs of tree drawings have not been the subject of research until now, with the exception of the recent work by Arseneva et al.~\cite{DBLP:conf/gd/ArsenevaBCDDFLT18}, who showed how to construct a three-dimensional crossing-free morph between two straight-line planar drawings of an $n$-node tree in $O(\log n)$ morphing steps. 

The rest of the paper is organized as follows. In \cref{se:preliminaries} we present some definitions and preliminaries. In \cref{se:upward} we present our results on small-area upward planar morphs between strictly-upward straight-line planar grid drawings of rooted trees. In \cref{se:general} we present our main result on small-area planar morphs between straight-line planar grid drawings of trees. Finally, in \cref{se:conclusions} we conclude, present some open problems, and argue about the generality of the model adopted in this paper. Namely, we prove that the problem of constructing a planar morph with polynomial resolution between two planar straight-line drawings of the same graph can be reduced to the problem of constructing a planar morph with polynomial area between two planar straight-line grid drawings of the same graph.

\section{Preliminaries}\label{se:preliminaries}

In this section we introduce some definitions and preliminaries; see also~\cite{DETT}. 

\paragraph{Trees.} The node and edge sets of a tree $T$ are denoted by $V(T)$ and $E(T)$, respectively. The \emph{degree} $\deg(v)$ of a node $v$ of $T$ is the number of its neighbors. In an \emph{ordered} tree, a counter-clockwise order of the edges incident to each node is specified. 

A \emph{rooted tree} $T$ is a tree with one distinguished node, which is called \emph{root} and is denoted by $r(T)$. For any node $u \in V(T)$ with $u\neq r(T)$, consider the unique path from $u$ to $r(T)$ in $T$; the \emph{ancestors} of $u$ are the nodes of such a path, the \emph{proper ancestors} of $u$ are the ancestors different from $u$ itself, and the \emph{parent} $p(u)$ of $u$ is the proper ancestor of $u$ which is adjacent to $u$. For any two nodes $u,v\in V(T)$, the \emph{lowest common ancestor} is the ancestor of $u$ and $v$ whose graph-theoretic distance from $r(T)$ is maximum. For any node $u\in V(T)$ with $u\neq r(T)$, the \emph{children} of $u$ are the neighbors of $u$ different from $p(u)$; the \emph{children} of $r(T)$ are all its neighbors. The nodes that have children are called \emph{internal}; a non-internal node is a \emph{leaf}.
For any node $u\in V(T)$ with $u\neq r(T)$, the \emph{subtree} $T_u$ of $T$ rooted at $u$ is defined as follows: remove from $T$ the edge $(u,p(u))$, thus separating $T$ in two trees; the one containing $u$ is the subtree of $T$ rooted at $u$. If each node of $T$ has at most two children, then $T$ is a \emph{binary tree}. 

An \emph{ordered rooted tree} is a tree that is rooted and ordered. In an ordered rooted tree $T$, for each node $u\in V(T)$, a \emph{left-to-right} (linear) order $u_1,\dots,u_k$ of the children of $u$ is specified.
If $T$ is binary then the first (second) child in the left-to-right order of the children of any node $u$ is the \emph{left} (\emph{right}) \emph{child} of $u$, and the subtree rooted at the left (right) child of $u$ is the \emph{left} (\emph{right}) \emph{subtree} of $u$.

\paragraph{Tree drawings.} In a \emph{straight-line drawing} $\Gamma$ of a tree $T$ each node $u$ is represented by a point of the plane (whose coordinates are denoted by $x_{\Gamma}(u)$ and $y_{\Gamma}(u)$) and each edge is represented by a straight-line segment between its end-points. All the drawings considered in this paper are straight-line, even when not specified. In a \emph{planar} drawing no two edges intersect except, possibly, at common end-points. For a rooted tree $T$, a \emph{strictly-upward} drawing $\Gamma$ is such that each edge $(u,p(u))\in E(T)$ is represented by a curve monotonically increasing in the $y$-direction from $u$ to $p(u)$; if $\Gamma$ is a straight-line drawing, this is equivalent to requiring that $y_{\Gamma}(u)<y_{\Gamma}(p(u))$. For an ordered tree $T$, an \emph{order-preserving} drawing $\Gamma$ is such that, for each node $u\in V(T)$, the counter-clockwise order of the edges incident to $u$ in $\Gamma$ is the same as the order associated to $u$ in $T$. Note that a strictly-upward drawing $\Gamma$ of an ordered rooted tree $T$ is order-preserving if and only if, for each node $u\in V(T)$, the edges from $u$ to its children enter $u$ in the left-to-right order associated with $u$. 

The \emph{bounding box} of a drawing $\Gamma$ is the smallest axis-parallel rectangle enclosing $\Gamma$. In a \emph{grid} drawing $\Gamma$ each node has integer coordinates; then the \emph{width} and the \emph{height} of $\Gamma$, denoted by $w(\Gamma)$ and $h(\Gamma)$, respectively, are the number of grid columns and rows intersecting the bounding box of $\Gamma$, while the \emph{area} of $\Gamma$ is its width times its height. 
For a node $v$ in a drawing $\Gamma$, an \emph{$\ell$-box centered at $v$} is the convex hull of the square whose corners are $(x_\Gamma(v) \pm \frac{\ell}{2}, y_\Gamma(v) \pm \frac{\ell}{2})$.

\paragraph{Morphs.} A \emph{morph} between two straight-line drawings $\Gamma_0$ and $\Gamma_1$ of a graph $G$ is a continuously changing family of drawings $\{\Gamma_t\}$ of $G$ indexed by time $t\in [0,1]$, such that the drawing at time $t=0$ is $\Gamma_0$ and the drawing at time $t=1$ is $\Gamma_1$. A morph is \emph{planar} if all its intermediate drawings are planar. A morph between two strictly-upward drawings of a rooted tree is \emph{upward} if all its intermediate drawings are strictly-upward.  A morph is \emph{linear} if each node moves along a straight-line trajectory at constant speed. Whenever the linear morph between two straight-line planar drawings $\Gamma_0$ and $\Gamma_1$ of a graph $G$ is not planar, one is usually interested in the construction of a piecewise-linear morph with small complexity between $\Gamma_0$ and $\Gamma_1$. This is formalized by defining a \emph{morph} between $\Gamma_0$ and $\Gamma_1$ as a sequence $\langle \Gamma_0=\Delta_0,\Delta_1,\dots,\Delta_k=\Gamma_1 \rangle$ of drawings of $G$ such that the linear morph $\langle \Delta_{i-1},\Delta_{i} \rangle$ is planar, for $i=1,\dots,k$; each linear morph $\langle \Delta_{i-1},\Delta_{i} \rangle$ is called a \emph{morphing step} or simply a \emph{step}.  

The \emph{width} $w(\mathcal M)$ of a morph $\mathcal M=\langle \Delta_0,\Delta_1,\dots,\Delta_k \rangle$, where $\Delta_i$ is a grid drawing, for $i=0,1,\dots,k$, is equal to $\max \{w(\Delta_0),w(\Delta_1),\dots,w(\Delta_k)\}$. The \emph{height} $h(\mathcal M)$ of $\mathcal M$ and the area of $\mathcal M$ are defined analogously.

The algorithms we design in this paper receive in input two order-preserving straight-line planar grid drawings $\Gamma_0$ and $\Gamma_1$ of an ordered tree and construct morphs $\langle \Gamma_0=\Delta_0,\Delta_1,\dots,\Delta_k=\Gamma_1 \rangle$ with few steps and small area.

\begin{remark}
	A necessary and sufficient condition for the existence of a planar morph between two straight-line planar drawings $\Gamma_0$ and $\Gamma_1$ of a tree $T$ is that they are ``topologically-equivalent'', i.e., the counter-clockwise order of the edges incident to each node $u\in V(T)$ is the same in $\Gamma_0$ and $\Gamma_1$. In order to better exploit standard terminology about tree drawings, we ensure that $\Gamma_0$ and $\Gamma_1$ are topologically-equivalent by assuming that $T$ is ordered and that $\Gamma_0$ and $\Gamma_1$ are order-preserving drawings; hence, dealing with ordered trees and with order-preserving drawings is not a loss of generality.
\end{remark}

\begin{remark}
	The width and the height of the morphs we construct are expressed not only in terms of the number of nodes of the input tree $T$, but also in terms of the width and the height of the input drawings $\Gamma_0$ and $\Gamma_1$ of $T$; this is necessary, given that $\max\{w(\Gamma_0),w(\Gamma_1)\}$ and  $\max\{h(\Gamma_0),h(\Gamma_1)\}$ are obvious lower bounds for the width and the height of any morph between $\Gamma_0$ and $\Gamma_1$, respectively.   
\end{remark}

\begin{remark}
	The morphs $\langle \Delta_0,\Delta_1,\dots,\Delta_k \rangle$ we construct in this paper are such that $\Delta_0,\Delta_1,\dots,\Delta_k$ are \emph{grid} drawings, even when not explicitly specified.
\end{remark}

In the following we introduce two tools we are going to use later. First, we observe that whether a linear morph is upward only depends on the upwardness of the initial and final drawings of the morph.

\begin{observation} \label{obs:upward-stays-upward}
	Let $\Gamma_0$ and $\Gamma_1$ be two strictly-upward straight-line drawings of a rooted tree $T$. Then the linear morph $\langle \Gamma_0,\Gamma_1 \rangle$ is upward.
\end{observation}	

\begin{proof}
	Assume that the morph $\langle \Gamma_0,\Gamma_1 \rangle$ happens between the time instants $t=0$ and $t=1$. For any $t\in [0,1]$, denote by $\Gamma_t$ the drawing of $T$ in $\langle \Gamma_0,\Gamma_1 \rangle$ at time $t$.
	
	Consider any edge $(p(v),v)$ of $T$. Since $\Gamma_0$ and $\Gamma_1$ are strictly-upward, it follows that $y_{\Gamma_0}(p(v))>y_{\Gamma_0}(v)$ and $y_{\Gamma_1}(p(v))>y_{\Gamma_1}(v)$. Hence, at any time instant $t\in [0,1]$ of the morph $\langle \Gamma_0,\Gamma_1 \rangle$, we have $y_{\Gamma_t}(p(v))=(1-t)\cdot y_{\Gamma_0}(p(v)) + t \cdot y_{\Gamma_1}(p(v)) > (1-t)\cdot  y_{\Gamma_0}(v) + t \cdot y_{\Gamma_1}(v) = y_{\Gamma_t}(v)$. It follows that the drawing $\Gamma_t$ is strictly-upward, and hence that $\langle \Gamma_0,\Gamma_1 \rangle$ is an upward morph.
\end{proof}

The planarity of a morph can not be ensured as simply as its upwardness. However, if the initial and final drawings of the morph satisfy some further conditions, it turns out that planarity is actually guaranteed; this is similar to a lemma by Da Lozzo {\em et al.}~\cite{DBLP:conf/gd/LozzoBFPR18}.

\begin{lemma}\label{le:unidirectional-morph}
	Let $\Gamma_0$ and $\Gamma_1$ be two order-preserving strictly-upward straight-line planar drawings of a rooted ordered tree $T$. Suppose that, for each node $v\in V(T)$, we have $y_{\Gamma_0}(v)=y_{\Gamma_1}(v)$. Then the linear morph $\langle \Gamma_0,\Gamma_1 \rangle$ is planar.
\end{lemma}

\begin{proof}
	The proof exploits Corollary 7.2~in~\cite{DBLP:journals/siamcomp/AlamdariABCLBFH17}, which we introduce in the following. 
	
	Assume that the morph $\langle \Gamma_0,\Gamma_1 \rangle$ happens between the time instants $t=0$ and $t=1$. For any $t\in [0,1]$, denote by $\Gamma_t$ the drawing of $T$ in $\langle \Gamma_0,\Gamma_1 \rangle$ at time $t$ and denote by $u_t$ the position of a node $u$ at time $t$, where $u_t=(1-t) \cdot u_0 + t \cdot u_1$.
	
	
	Now, consider a point $q_0$ \emph{of} $\Gamma_0$, that is a point that represents a node of $T$ or that belongs to a segment representing an edge of $T$ in $\Gamma_0$. The morph $\langle \Gamma_0,\Gamma_1 \rangle$ moves $q_0$ to a point $q_1$ of $\Gamma_1$. This is evident if $q_0$ represents a node of $T$ in $\Gamma_0$; if $q_0$ is a point of the segment representing an edge $(u,v)$ in $\Gamma_0$, then $q_0=(1-\gamma) \cdot u_0+ \gamma \cdot v_0$, for some value $0<\gamma<1$, and the point $q_1=(1-\gamma) \cdot u_1 + \gamma \cdot v_1$ indeed belongs to the segment representing $(u,v)$ in $\Gamma_1$. More in general, for any $t\in [0,1]$, the point $q_t=(1-t) \cdot q_0 + t \cdot q_1$ belongs to the segment representing $(u,v)$ in $\Gamma_t$.
	
	Now consider any three points $q_0$, $r_0$, and $s_0$ of $\Gamma_0$, which are moved to three points $q_1$, $r_1$, and $s_1$ of $\Gamma_1$ by $\langle \Gamma_0,\Gamma_1 \rangle$, respectively. Suppose that $q_i$ is to the left (to the right) of the  line through $r_i$ and $s_i$, for each $i=0,1$. Then Corollary 7.2~in~\cite{DBLP:journals/siamcomp/AlamdariABCLBFH17} ensures that $q_t$ is to the left (resp.\ to the right) of the line through $r_t$ and $s_t$, for any $t\in [0,1]$. The applicability of Corollary 7.2~from~\cite{DBLP:journals/siamcomp/AlamdariABCLBFH17} to the morph $\langle \Gamma_0,\Gamma_1 \rangle$ is a consequence of the fact that $\langle \Gamma_0,\Gamma_1 \rangle$ is \emph{unidirectional}, that is, all the node trajectories are parallel (indeed, each node $v\in V(T)$ moves along a horizontal line, given that $y_{\Gamma_0}(v)=y_{\Gamma_1}(v)$). 
	
	We prove the planarity of $\langle \Gamma_0,\Gamma_1 \rangle$. Suppose, for a contradiction, that two edges $(u,p(u))$ and $(v,p(v))$ of $T$ cross during $\langle \Gamma_0,\Gamma_1 \rangle$. For the sake of the simplicity of notation, let $u'=p(u)$ and $v'=p(v)$. Then there exists a point $q_0$ that belongs the segment $\overline{u_0u'_0}$ and such that, for some $t\in (0,1)$, the point $q_{t}=(1-{t}) \cdot q_0 + {t} \cdot q_1$ is a point of the segment $\overline{u_{t}u'_{t}}$ and also a point of the segment $\overline{v_{t} v'_{t}}$. Since every node of $T$ only moves horizontally in $\langle \Gamma_0,\Gamma_1 \rangle$, we have that $q_0$ ($q_1$) lies in the strip delimited by the horizontal lines through $v_0$ and $v'_0$ (resp.\ through $v_1$ and $v'_1$). Since $\Gamma_0$ and $\Gamma_1$ are planar, $q_0$ and $q_1$ do not belong to the segments $\overline{v_0v'_0}$ and $\overline{v_1v'_1}$, respectively. It follows that $q_0$ is either to the left or to the right of the line $\ell_0^v$ that passes through $v_0$ and $v'_0$; analogously, $q_1$ is either to the left or to the right of the line $\ell_1^v$ that passes through $v_1$ and $v'_1$. Assume that $q_0$ is to the left of $\ell_0^v$; we claim that this implies that $q_1$ is to the left of $\ell_1^v$. The claim follows from the fact that, both in $\Gamma_0$ and in $\Gamma_1$, the two paths of $T$ from $u$ and $v$ to their lowest common ancestor $w$ are monotone in the $y$-direction (since $\Gamma_0$ and $\Gamma_1$ are strictly-upward), do not cross each other (since $\Gamma_0$ and $\Gamma_1$ are planar), and enter $w$ in the same left-to-right order (since $\Gamma_0$ and $\Gamma_1$ are order-preserving). Now since $q_0$ is to the left of $\ell_0^v$ and $q_1$ is to the left of $\ell_1^v$, Corollary 7.2~from~\cite{DBLP:journals/siamcomp/AlamdariABCLBFH17} implies that $q_{t}$ is to the left of the line that passes through $v_{t}$ and $v'_{t}$, while $q_{t}$ is supposed to be a point of $\overline{v_{t} v'_{t}}$. This contradiction proves the lemma.
\end{proof}

Upward planar morphs between strictly-upward drawings of rooted ordered trees maintain the drawing order-preserving at all times, as proved in the following.

\begin{lemma}\label{le:upwardplanar-implies-order}
	Let $\Gamma_0$ and $\Gamma_1$ be two order-preserving strictly-upward straight-line planar drawings of a rooted ordered tree $T$. Let $\mathcal M$ be any upward planar morph between $\Gamma_0$ and $\Gamma_1$. Then any intermediate drawing of $\mathcal M$ is order-preserving.
\end{lemma}

\begin{proof}
	Assume that the morph $\mathcal M$ happens between the time instants $t=0$ and $t=1$. For any $t\in [0,1]$, denote by $\Gamma_t$ the drawing of $T$ in $\mathcal M$ at time $t$. Since $\mathcal M$ is upward, the drawing $\Gamma_t$ is strictly-upward, for any $t\in [0,1]$. Hence, it suffices to prove that, for any internal node $v$ of $T$, the edges from $v$ to its children enter $v$ in $\Gamma_t$ in the left-to-right order associated with $v$; this is indeed true for $t=0$ and $t=1$. Suppose that, for some $t\in (0,1)$, the left-to-right order in which two edges $(u_1,v)$ and $(u_2,v)$ enter $v$ in $\Gamma_t$ is different than in $\Gamma_0$. Since such edges are represented by curves monotonically increasing in the $y$-direction from $u_1$ and $u_2$ to $v$ throughout $\mathcal M$, it follows that there is a time $t^*\in (0,t)$ such that the edges $(u_1,v)$ and $(u_2,v)$ overlap in $\Gamma_{t^*}$. However, this cannot happen due to the planarity of  $\mathcal M$.
\end{proof}	 
	

\section{Upward Planar Morphs of Rooted-Tree Drawings}\label{se:upward}

\newcommand{\osusp}{order-preserving strictly-upward straight-line planar grid\xspace}
In this section we show how to construct small-area morphs between order-preserving strictly-upward straight-line planar grid drawings of rooted ordered~trees.

Our first result shows that such morphs can always be constructed consisting of a linear number of steps. This is obtained via an inductive algorithm which is described in the following.
Let $T$ be an $n$-node rooted ordered tree. The \emph{rightmost path} of $T$ is the maximal path $(s_0,\dots,s_m)$ such that $s_0=r(T)$ and $s_{i}$ is the rightmost child of $s_{i-1}$, for $i=1,\dots,m$. Note that $s_m$ is a leaf, which is called the \emph{rightmost leaf} $l^{\rightarrow}_T$ of $T$. For a straight-line grid drawing $\Gamma$, denote by $\ell_{\Gamma}$ the rightmost vertical line intersecting $\Gamma$; note that $\ell_{\Gamma}$ is a grid column. 

Let $\Gamma_0$ and $\Gamma_1$ be two \osusp drawings of $T$. We inductively construct a morph $\mathcal M$ from $\Gamma_0$ to $\Gamma_1$~as~follows. 

In the base case $n=1$; then $\mathcal M$ is the linear morph $\langle \Gamma_0,\Gamma_1\rangle$.

In the inductive case $n>1$. Let $l=l^{\rightarrow}_T$ be the rightmost leaf of $T$. Let $\pi=p(l)$ be the parent of $l$. Let $T'$ be the $(n-1)$-node tree obtained from $T$ by removing the node $l$ and the edge $(\pi,l)$. Let $\Gamma'_0$ and $\Gamma'_1$ be the drawings of $T'$ obtained from $\Gamma_0$ and $\Gamma_1$, respectively, by removing the node $l$ and the edge $(\pi,l)$. Inductively compute a $k$-step upward planar morph $\mathcal M'=\langle \Gamma'_0=\Delta'_1,\Delta'_2,\dots,\Delta'_{k}=\Gamma'_1\rangle$.

We now construct a morph $\mathcal M=\langle \Gamma_0,\Delta_1,\Delta_2,\dots,\Delta_k,\Gamma_1\rangle$. For each $i=2,3,\dots,k-1$, we define $\Delta_i$ as the drawing obtained from $\Delta'_i$ by placing $l$ one unit below $\pi$ and one unit to the right of $\ell_{\Delta'_i}$. Further, we define $\Delta_1$ ($\Delta_{k}$) as the drawing obtained from $\Delta'_1$ (resp.\ from $\Delta'_k$) by placing $l$ one unit below $\pi$ and one unit to the right of $\ell_{\Gamma_0}$ (resp.\ $\ell_{\Gamma_1}$). Note that the point at which $l$ is placed in $\Delta_1$ (in $\Delta_k$) is one unit to the right of $\ell_{\Delta'_1}$ (resp.\ $\ell_{\Delta'_k}$), similarly as in $\Delta_2, \Delta_3, \dots, \Delta_{k-1}$, except if $l$ is to the right of every other node of $\Gamma_0$ (of $\Gamma_1$); in that case $l$ might be several units to the right of $\ell_{\Delta'_1}$ (resp.\ $\ell_{\Delta'_k}$). This completes the construction of $\mathcal M$. We get the following.



%

\begin{theorem}\label{th:nary-trees-linearsteps-polyarea}
	Let $T$ be an $n$-node rooted ordered tree, and let $\Gamma_0$ and $\Gamma_1$ be two \osusp drawings of $T$. 
	There exists a $(2n-1)$-step upward planar morph $\mathcal M$ from $\Gamma_0$ to $\Gamma_1$ 
	with $h(\mathcal M) = \max\{h(\Gamma_0),h(\Gamma_1)\}$ and $w(\mathcal M) = \max\{w(\Gamma_0),w(\Gamma_1)\} +n-1$.
\end{theorem}

\begin{proof}
For a drawing $\Gamma$ of $T$ and for any two nodes $s$ and $t$ of $T$, denote by $d^v_\Gamma(s,t)$ the vertical distance between $s$ and $t$ in $\Gamma$ (that is, the absolute value of the difference between their $y$-coordinates).

We claim that the morph $\mathcal M=\langle \Gamma_0,\Delta_1,\Delta_2,\dots,\Delta_k,\Gamma_1\rangle$ defined before the statement of the theorem satisfies the requirements of the theorem and also satisfies the property that $d^v_{\Delta_i}(u,r(T))\leq \max \{d^v_{\Gamma_0}(u,r(T)),d^v_{\Gamma_1}(u,r(T))\}$, for any $i=1,2,\dots,k$ and for each node $u$ of $T$. 

First, note that $\mathcal M$ has $k+1=2n-1$ steps. This is trivially true if $n=1$ and it is true by induction if $n>1$ since $\mathcal M$ has $2$ steps more than $\mathcal M'$.	

For each $i=1,2,\dots,k$, the drawing $\Delta_i$ is straight-line by construction; further, $\Delta_i$ is strictly-upward since $\Delta'_i$ is strictly-upward and since $l$ lies one unit below $\pi$ in $\Delta_i$. By \cref{obs:upward-stays-upward}, the morph $\mathcal M$ is upward. Since $\ell_{\Gamma_0}, \ell_{\Delta'_2},\ell_{\Delta'_3},\dots,\ell_{\Delta'_{k-1}},\ell_{\Gamma_1}$ are grid columns and since $\pi$ is placed at a grid point in each of $\Delta'_1,\Delta'_2,\dots,\Delta'_k$, we have that $l$ is placed at a grid point in each of $\Delta_1,\Delta_2,\dots,\Delta_k$, hence each of such drawings is a grid drawing.
 
We now analyze $w(\mathcal M)$. If $n=1$, then $\mathcal M=\langle \Gamma_0,\Gamma_1\rangle$, hence $w(\mathcal M)=\max\{w(\Gamma_0),w(\Gamma_1)\} +n-1$. Assume next that $n\geq 2$. Consider any $i\in \{2,3,\dots,k-1\}$. By construction $\Delta_i$ occupies one more grid column than $\Delta'_i$. By induction $w(\Delta'_i)\leq \max\{w(\Gamma'_0),w(\Gamma'_1)\} +n-2$; further, $w(\Gamma'_0)\leq w(\Gamma_0)$ and $w(\Gamma'_1)\leq w(\Gamma_1)$. Hence, $w(\Delta_i)\leq  \max\{w(\Gamma_0),w(\Gamma_1)\} +n-1$, as required. In order to bound $w(\Delta_1)$, the argument is the same as the one above if $\ell_{\Gamma_0}$ contains a node of $\Delta'_1$. Otherwise, $\ell_{\Gamma_0}$ contains $l$ and no other node of $\Gamma_0$; then, by construction, we have $w(\Delta_1)= w(\Gamma_0)+1\leq \max\{w(\Gamma_0),w(\Gamma_1)\} +n-1$, as required. The proof that $w(\Delta_k)\leq \max\{w(\Gamma_0),w(\Gamma_1)\} +n-1$ is analogous.
	
Next, consider any index $i\in \{1,2,\dots,k\}$ and any node $u$ of $T$. We prove that $d^v_{\Delta_i}(u,r(T))\leq \max \{d^v_{\Gamma_0}(u,r(T)),d^v_{\Gamma_1}(u,r(T))\}$.

\begin{itemize}
	\item If $u\neq l$, then we have $d^v_{\Delta_i}(u,r(T))= d^v_{\Delta'_i}(u,r(T'))$. By induction, we have $d^v_{\Delta'_i}(u,r(T'))\leq \max \{d^v_{\Gamma'_0}(u,r(T')),d^v_{\Gamma'_1}(u,r(T'))\}$. Since $d^v_{\Gamma'_0}(u,r(T'))= d^v_{\Gamma_0}(u,r(T))$ and $d^v_{\Gamma'_1}(u,r(T'))= d^v_{\Gamma_1}(u,r(T))$, we have $d^v_{\Delta_i}(u,r(T))\leq \max \{d^v_{\Gamma_0}(u,r(T)),d^v_{\Gamma_1}(u,r(T))\}$ as required. 
	\item By construction, we have $d^v_{\Delta_i}(l,r(T))= d^v_{\Delta_i}(\pi,r(T))+1$. Further, $d^v_{\Delta_i}(\pi,r(T))=d^v_{\Delta'_i}(\pi,r(T'))$. By induction we have $d^v_{\Delta'_i}(\pi,r(T'))\leq \max \{d^v_{\Gamma'_0}(\pi,r(T')),d^v_{\Gamma'_1}(\pi,r(T'))\}$. Since $d^v_{\Gamma'_0}(\pi,r(T'))=d^v_{\Gamma_0}(\pi,r(T))$ and $d^v_{\Gamma'_1}(\pi,r(T'))=d^v_{\Gamma_1}(\pi,r(T))$, we have $d^v_{\Delta_i}(l,r(T)) \leq \max \{d^v_{\Gamma_0}(\pi,r(T)),d^v_{\Gamma_1}(\pi,r(T))\}+1$. Since $\Gamma_0$ and $\Gamma_1$ are strictly-upward, we have $d^v_{\Gamma_0}(l,r(T))\geq d^v_{\Gamma_0}(\pi,r(T))+1$ and $d^v_{\Gamma_1}(l,r(T))\geq d^v_{\Gamma_1}(\pi,r(T))+1$, hence $d^v_{\Delta_i}(l,r(T)) \leq \max \{d^v_{\Gamma_0}(l,r(T)),d^v_{\Gamma_1}(l,r(T))\}$ as required.
\end{itemize}

Since $d^v_{\Delta_i}(u,r(T))\leq \max \{d^v_{\Gamma_0}(u,r(T)),d^v_{\Gamma_1}(u,r(T))\}$, for any index $i\in \{1,2,\dots,k\}$ and any node $u$ of $T$, we directly get $h(\Delta_i)\leq \max\{h(\Gamma_0),h(\Gamma_1)\}$ and hence $h(\mathcal M) = \max\{h(\Gamma_0),h(\Gamma_1)\}$.

It remains to prove the planarity of $\mathcal M$; then Lemma~\ref{le:upwardplanar-implies-order} implies that the drawing of $T$ is order-preserving throughout $\mathcal M$. The following property is useful (refer to \cref{fig:upward-linearsteps-polyarea-property}). 

\begin{figure}[tb]
	\centering
	\subfloat[]{\includegraphics[]{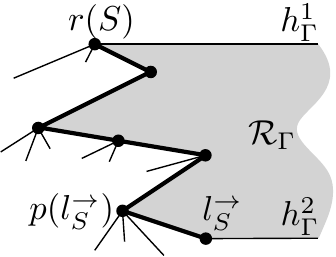}\label{fig:upward-linearsteps-polyarea-property}}\hfil
	\subfloat[]{\includegraphics[]{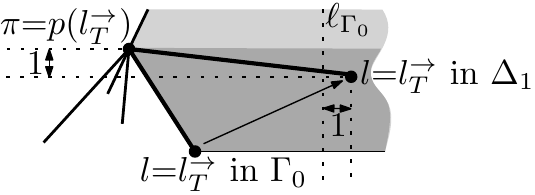}\label{fig:upward-linearsteps-polyarea-planar}}\\
	\subfloat[]{\includegraphics[]{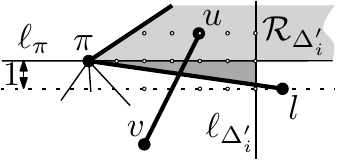}\label{fig:upward-linearsteps-intermediate}}\hfil
	\subfloat[]{\includegraphics[]{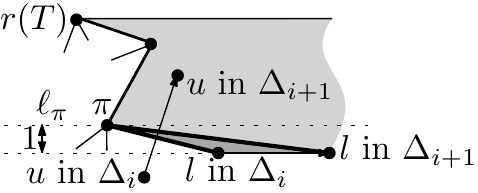}\label{fig:upward-linearsteps-intermediate-new}}
	\caption{(a) Illustration for \cref{pr:empty-right}. The rightmost path $R_S$ of $S$ is represented by a thick line. (b) Illustration for the proof of planarity of $\langle \Gamma_0,\Delta_1\rangle$. The gray region is $\mathcal R_{\Gamma_0}$; the darker gray region is $\Pi_{\Gamma_0}$. (c) Illustration for the proof of the planarity of $\Delta_i$. The gray region is $\mathcal R_{\Delta'_i}$; the darker gray region is the triangle $D$. (d) Illustration for the proof of the planarity of $\langle \Delta_i,\Delta_{i+1}\rangle$.}
\end{figure}

\begin{property} \label{pr:empty-right}
Consider any \osusp drawing $\Gamma$ of an ordered rooted tree $S$. Let $h^1_{\Gamma}$ and $h^2_{\Gamma}$ be the horizontal half-lines starting at the root $r(S)$ and at the rightmost leaf $l^{\rightarrow}_S$ of $S$, respectively, and directed rightwards. Let $\mathcal R_{\Gamma}$ be the region of the plane which is delimited by the rightmost path $R_S$ of $S$ from the left and by $h^1_{\Gamma}$ and $h^2_{\Gamma}$ from above and below, respectively. Then no node or edge of $S$, other than those of $R_S$,~intersects~$\mathcal R_{\Gamma}$. 
\end{property}

\begin{proof}
	We first argue about possible intersections with the interior of $\mathcal R_{\Gamma}$. 
	
	\begin{itemize}
		\item Suppose, for a contradiction, that an edge of $S$ intersects the interior of $\mathcal R_{\Gamma}$. Since $S$ is connected, it follows that there is an edge $e$ of $S$ that intersects the interior of $\mathcal R_{\Gamma}$ {\em and} its boundary. Since $\Gamma$ is planar, the edge $e$ does not cross $R_S$. Since $R_S$ is the rightmost path of $S$, the edge $e$ does not share a node with $R_S$. Since $\Gamma$ is strictly-upward and $r(S)$ is the root of $S$, we have that $e$ does not intersect $h^1_{\Gamma}$ other than, possibly, at $r(S)$; however, $r(S)$ belongs to $R_S$, and we already ruled out the possibility that $e$ shares a node with $R_S$. Finally, if $e$ intersects $\mathcal R_{\Gamma}$ and $h^2_{\Gamma}$, then, since $\Gamma$ is strictly-upward, the path from $e$ to $r(S)$ intersects $R_S$, hence there is an edge $e'$ of $S$ that intersects $\mathcal R_{\Gamma}$ and $R_S$, a case which we already ruled out; we thus get a contradiction.
		\item Suppose, for a contradiction, that a node of $S$ lies in the interior of $\mathcal R_{\Gamma}$. Since $S$ is connected, it follows that an edge of $S$ also intersects the interior of $\mathcal R_{\Gamma}$, a case which we already ruled out; we thus get a contradiction. 
	\end{itemize}
	
	We now argue about possible intersections with the boundary of $\mathcal R_{\Gamma}$. 
	
	\begin{itemize}
		\item Suppose, for a contradiction, that a node $u$ not in $R_S$ lies on the boundary of $\mathcal R_{\Gamma}$. Since $\Gamma$ is planar, the node $u$ does not lie on $R_S$. Since $\Gamma$ is strictly-upward and the root $r(S)$ of $S$ is a node of $R_S$, the node $u$ does not lie on $h^1_{\Gamma}$. Finally, if $u$ lies on $h^2_{\Gamma}$, then since $\Gamma$ is strictly-upward, the edge from $u$ to its parent intersects the interior of $\mathcal R_{\Gamma}$, a case which we already ruled out; we thus get a contradiction.
		\item Suppose, for a contradiction, that the interior of an edge $e$ not in $R_S$ intersects the boundary $\mathcal R_{\Gamma}$. Since $\Gamma$ is planar, the interior of $e$ does not intersect $R_S$. If the interior of $e$ intersects $h^1_{\Gamma}$ or $h^2_{\Gamma}$, then, since $\Gamma$ is strictly-upward, it intersects the interior of $\mathcal R_{\Gamma}$ as well, a case which we already ruled out; we thus get a contradiction.
	\end{itemize}
	This concludes the proof.
\end{proof}

We now exploit \cref{pr:empty-right} to prove the planarity of $\mathcal M$. Since $\mathcal M'$ is planar, by induction, and since $\mathcal M$ coincides with $\mathcal M'$ when restricted to the nodes and edges of $T'$, it follows that $\mathcal M$ is planar, as long as the edge $(\pi,l)$ does not intersect any edge, other than at a common end-point, throughout $\mathcal M$. We now argue about the possible intersections of the edge $(\pi,l)$ in $\mathcal M$. 

First, we deal with the morph $\langle \Gamma_0,\Delta_1 \rangle$, in which only the node $l$ moves. By \cref{pr:empty-right} applied to $T$ and $\Gamma_0$, no node or edge of $T$, other than those of the righmost path $R_T$ of $T$, intersects $\mathcal R_{\Gamma_0}$. Hence, it suffices to prove that the edge $(\pi,l)$ lies in $\mathcal R_{\Gamma_0}$ throughout $\langle \Gamma_0,\Delta_1 \rangle$; refer to \cref{fig:upward-linearsteps-polyarea-planar}. Since $\Gamma_0$ is a strictly-upward grid drawing, we have that $l$ lies at least one unit below $\pi$ in $\Gamma_0$; further, by construction, $l$ lies one unit below $\pi$ in $\Delta_1$. Moreover, by construction, the position of $l$ in $\Delta_1$ is at least one unit to the right of the positions of $\pi$ and $l$ in $\Gamma_0$. It follows that the point at which $l$ is placed in $\Delta_1$ lies in the part $\Pi_{\Gamma_0}$ of $\mathcal R_{\Gamma_0}$ that is delimited by the representation of the edge $(\pi,l)$ in $\Gamma_0$ from the left and by the horizontal lines through the positions of $\pi$ and $l$ in $\Gamma_0$ from above and from below, respectively. The convexity of $\Pi_{\Gamma_0}$ implies that the edge $(\pi,l)$ lies in $\mathcal R_{\Gamma_0}$ throughout $\langle \Gamma_0,\Delta_1 \rangle$.

The proof that the  morph $\langle \Delta_k,\Gamma_1 \rangle$ is planar is symmetric.

We now prove that the drawing $\Delta_i$ is planar, for each $i=0,\dots,k$; refer to \cref{fig:upward-linearsteps-intermediate}. Since $\Delta'_i$ is planar, by induction, any crossing in $\Delta_i$ involves the edge $(\pi,l)$. Suppose, for a contradiction, that an edge $(u,v)$ crosses $(\pi,l)$. If $(u,v)$ is an edge of $R_{T}\setminus \{(\pi,l)\}$, then $(u,v)$ and $(\pi,l)$ are separated by the horizontal line through $\pi$ and thus do not cross. Assume hence that $(u,v)$ is not an edge of $R_{T}\setminus \{(\pi,l)\}$. The intersection between $(u,v)$ and $(\pi,l)$ in $\Delta_i$ has to happen at an interior point $c$ of both edges. Indeed:

\begin{itemize}
	\item $l$ is to the right of both $u$ and $v$, hence it is not on $(u,v)$; 
	\item $\pi$ is not in the interior of $(u,v)$, since $\Delta'_i$ is planar; 
	\item $u$ and $v$ are not in the interior of $(\pi,l)$, given that $y_{\Delta_i}(l)=y_{\Delta_i}(\pi)-1$ and hence there is no grid point in the interior of $(\pi,l)$;
	\item if one of $u$ and $v$, say $u$, overlaps with $\pi$, then the planarity of $\Delta'_i$ implies that $u$ is the same node as $\pi$; since $v$ is not in the interior of $(\pi,l)$ and $l$ is not on $(u,v)$, it follows that $(u,v)$ \mbox{and $(\pi,l)$ do not cross.}
\end{itemize} 

Since the interior of $(u,v)$ intersects the interior of $(\pi,l)$, it follows that $(u,v)$ intersects the interior of the triangle $D$ which is delimited by $(\pi,l)$, by the horizontal line $\ell_\pi$ through $\pi$, and by $\ell_{\Delta'_i}$. Since $y_{\Delta_i}(l)=y_{\Delta_i}(\pi)-1$, it follows that $D$ contains no grid point in its interior. Hence, either $(u,v)$ intersects $\ell_{\pi}$ to the right of $\pi$, thus contradicting \cref{pr:empty-right} for $\Delta'_i$, or $(u,v)$ intersects $\ell_{\Delta'_i}$ below $\ell_{\pi}$, thus implying that $u$ or $v$ is to the right of $\ell_{\Delta'_i}$; this contradicts the definition of $\ell_{\Delta'_i}$ and hence proves that $\Delta_i$ is planar.

Finally, we prove that, for each $i=1,\dots,k-1$, the morph $\langle \Delta_i,\Delta_{i+1} \rangle$ is planar. Suppose, for a contradiction, that an edge $(u,v)$ of $T$ crosses the edge $(\pi,l)$ during $\langle \Delta_i,\Delta_{i+1} \rangle$. Consider the first drawing during $\langle \Delta_i,\Delta_{i+1} \rangle$ in which $(u,v)$ and $(\pi,l)$ cross; denote such a drawing by $\Gamma$ and recall that $\Gamma\neq \Delta_i,\Delta_{i+1}$. Since any drawing before $\Gamma$ in $\langle \Delta_i,\Delta_{i+1} \rangle$ is planar, it follows that in $\Gamma$ an end-point of one of $(u,v)$ and $(\pi,l)$ lies on the other edge. Since $l$ is to the right of $u$ and $v$ both in $\Delta_i$ and in $\Delta_{i+1}$, it follows that $l$ is to the right of $u$ and $v$ throughout $\langle \Delta_i,\Delta_{i+1} \rangle$, hence it does not lie on $(u,v)$ in $\Gamma$. Since $\langle \Delta'_i,\Delta'_{i+1} \rangle$ is planar, $\pi$ does not lie on $(u,v)$ (except if it is the same node as $u$ or $v$, which however does not cause the crossing between $(u,v)$ and $(\pi,l)$). Hence, one of $u$ and $v$, say $u$, lies in the interior of $(\pi,l)$ in $\Gamma$.

Assume that $\pi$ does not move during $\langle \Delta_i,\Delta_{i+1} \rangle$; this is not a loss of generality, as a planar linear morph between two drawings remains planar if one of the two drawings is translated by an arbitrary vector (see, e.g.,~\cite{DBLP:journals/siamcomp/AlamdariABCLBFH17}). Refer to \cref{fig:upward-linearsteps-intermediate-new}. Since $l$ is one unit below the horizontal line $\ell_{\pi}$ through $\pi$ both in $\Delta_i$ and in $\Delta_{i+1}$, it follows that $l$ moves horizontally in $\langle \Delta_i,\Delta_{i+1} \rangle$. By assumption, the straight-line segment representing the trajectory of $u$ in $\langle \Delta_i,\Delta_{i+1} \rangle$ crosses the triangle whose vertices are $\pi$ and the positions of $l$ in $\Delta_i$ and $\Delta_{i+1}$. Note that $\langle \Delta_i,\Delta_{i+1} \rangle$ moves $u$ from a grid point in $\Delta_i$ to a grid point in $\Delta_{i+1}$. However, since there are no grid points below $\ell_{\pi}$ and above the horizontal line through $l$, it follows that during $\langle \Delta_i,\Delta_{i+1} \rangle$  either $u$ crosses $R_{T'}$, a contradiction to the planarity of $\langle \Delta'_i,\Delta'_{i+1} \rangle$, or it crosses the horizontal line through $r(T)$, a contradiction to the fact that $\langle \Delta_i,\Delta_{i+1} \rangle$ is an upward morph, or it crosses the vertical line through $l$, a contradiction to the fact that $l$ is to the right of $u$ throughout $\langle \Delta_i,\Delta_{i+1} \rangle$, or it lies in $\mathcal R_{\Delta_i}$ in $\Delta_i$, or it lies in $\mathcal R_{\Delta_{i+1}}$ in $\Delta_{i+1}$; both the last two possibilities contradict \cref{pr:empty-right}.
This concludes the proof of the theorem.
\end{proof}

In view of \cref{th:nary-trees-linearsteps-polyarea}, it is natural to ask whether a sub-linear number of steps suffices to construct a small-area morph between any two \osusp drawings of a rooted ordered tree. In the following we prove that this is indeed the case for binary trees, for which just three morphing steps are sufficient. 

Our algorithm borrows ideas from a recent paper by Da Lozzo {\em et al.}~\cite{DBLP:conf/gd/LozzoBFPR18}, which deals with upward planar morphs of \emph{upward plane graphs}. 

Consider any two order-preserving strictly-upward straight-line planar grid drawings $\Gamma_0$ and $\Gamma_1$ of an $n$-node rooted ordered binary tree $T$. We define two order-preserving strictly-upward straight-line planar grid drawings $\Gamma'_0$ and $\Gamma'_1$ of $T$ such that the $3$-step morph $\langle \Gamma_0,\Gamma'_0,\Gamma'_1,\Gamma_1 \rangle$ is upward and planar. 

\begin{figure}[tb!]
	\centering
	\includegraphics[scale=0.6]{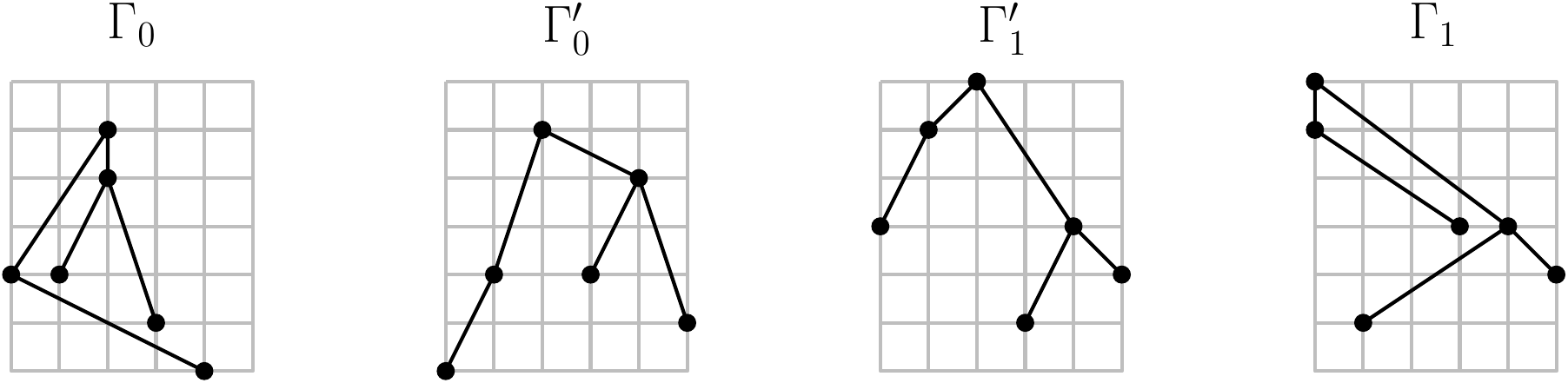}
	\caption{The $3$-step morph $\langle \Gamma_0,\Gamma'_0,\Gamma'_1,\Gamma_1 \rangle$.}
	\label{fig:binary-upward-morph}
\end{figure}

For $i=0,1$, we define $\Gamma'_i$ recursively as follows; refer to \cref{fig:binary-upward-morph}. Let $x_{\Gamma'_i}(r(T))=0$ and let $y_{\Gamma'_i}(r(T))=y_{\Gamma_i}(r(T))$. If the left subtree $L$ of $r(T)$ is non-empty, then recursively construct a drawing of it. Let $x_M$ be the maximum $x$-coordinate of a node in the constructed drawing of $L$; horizontally translate such a drawing by subtracting $x_M+1$ from the $x$-coordinate of every node in $L$, so that the maximum $x$-coordinate of any node in $L$ is now $-1$. Symmetrically, if the right subtree $R$ of $r(T)$ is non-empty, then recursively construct a drawing of it. Let $x_m$ be the minimum $x$-coordinate of a node in the constructed drawing of $R$; horizontally translate such a drawing by subtracting $x_m-1$ from the $x$-coordinate of every node in $R$, so that the minimum $x$-coordinate of any node in $R$ is now $1$.  

\begin{theorem} \label{th:binary-trees-upward-lineararea}
	Let $T$ be an $n$-node rooted ordered binary tree, and let $\Gamma_0$ and $\Gamma_1$ be two \osusp drawings of $T$. 
	There exists a $3$-step upward planar morph $\mathcal M$ from $\Gamma_0$ to $\Gamma_1$ 
	with $h(\mathcal M) = \max\{h(\Gamma_0),h(\Gamma_1)\}$ and $w(\mathcal M) = \max\{w(\Gamma_0),w(\Gamma_1),n\}$.
\end{theorem}

\begin{proof}
We prove that the morph $\mathcal M=\langle \Gamma_0,\Gamma'_0,\Gamma'_1,\Gamma_1\rangle$, where $\Gamma'_0$ and $\Gamma'_1$ are the drawings defined before the statement of the theorem, \mbox{satisfies the requirements.}

Consider any $i\in \{1,2\}$. By construction, we have $y_{\Gamma'_i}(r(T))=y_{\Gamma_i}(r(T))$; since the drawings of $L$ and $R$ are  constructed recursively and then translated horizontally, we have $y_{\Gamma'_i}(v)=y_{\Gamma_i}(v)$ for each node $v\in V(T)$. This has two implications. First, we have that $h(\Gamma'_i)=h(\Gamma_i)$ and hence that $h(\mathcal M) = \max\{h(\Gamma_0),h(\Gamma_1)\}$. Second, we have that $\Gamma'_i$ is strictly-upward, given that $\Gamma_i$ is strictly-upward, and hence that $\mathcal M$ is upward, by \cref{obs:upward-stays-upward}.

An easy inductive argument shows that no two nodes have the same $x$-coordinate in $\Gamma'_i$ and that every grid column intersecting $\Gamma'_i$ contains a node of $T$, hence $w(\Gamma'_i)=n$ and $w(\mathcal M) = \max\{w(\Gamma_0),w(\Gamma_1),n\}$. 


It remains to prove that $\mathcal M$ is planar; then Lemma~\ref{le:upwardplanar-implies-order} implies that the drawing of $T$ is order-preserving throughout $\mathcal M$.

We first deal with the planarity of the drawings $\Gamma'_0$ and $\Gamma'_1$ and of the morph $\langle \Gamma'_0,\Gamma'_1\rangle$. Note that the assignment of $x$-coordinates to the nodes of $T$ in $\Gamma'_0$ and in $\Gamma'_1$ depends on $T$, and not on $\Gamma_0$ or $\Gamma_1$. It follows that $x_{\Gamma'_0}(v)=x_{\Gamma'_1}(v)$, for each node $v\in V(T)$. Hence, each node moves along a vertical line in $\langle \Gamma'_0,\Gamma'_1\rangle$. We now prove that any two distinct edges $(u,p(u))$ and $(v,p(v))$ of $T$ do not cross in $\langle \Gamma'_0,\Gamma'_1\rangle$. If $u=p(v)$, then the edges $(u,p(u))$ and $(v,p(v))$ are separated by the horizontal line through $u$ throughout $\langle \Gamma'_0,\Gamma'_1\rangle$, given that such a morph is upward, hence they do not intersect, except at $u$; similarly, if $u$ is a proper ancestor of $p(v)$, or if $v=p(u)$, or if $v$ is a proper ancestor of $p(u)$, then the edges $(u,p(u))$ and $(v,p(v))$ do not cross. Now suppose that $u$ is not an ancestor of $p(v)$ and  $v$ is not an ancestor of $p(u)$. Then $u$ and $v$ are respectively in the left subtree and in the right subtree of their lowest common ancestor $w$ (up to renaming $u$ with $v$). By construction, the drawing of the left subtree of $r(T)$ lies to the left of the vertical line $\ell_w$ through $w$ both in $\Gamma'_0$ and in $\Gamma'_1$, hence it lies to the left of $\ell_w$ throughout $\langle \Gamma'_0,\Gamma'_1\rangle$. Analogously, the drawing of the right subtree of $r(T)$ lies to the right of $\ell_w$ throughout $\langle \Gamma'_0,\Gamma'_1\rangle$. It follows that $(u,p(u))$ and $(v,p(v))$ are separated by $\ell_w$ throughout $\langle \Gamma'_0,\Gamma'_1\rangle$, hence they do not cross. 

Since $\Gamma'_0$ and $\Gamma'_1$ are order-preserving strictly-upward straight-line planar drawings of $T$ and since $y_{\Gamma'_0}(v)=y_{\Gamma_0}(v)$ and $y_{\Gamma'_1}(v)=y_{\Gamma_1}(v)$ for each node $v\in V(T)$, \cref{le:unidirectional-morph} applies twice to ensure the planarity of the morphs $\langle \Gamma_0,\Gamma'_0\rangle$ and $\langle \Gamma'_1,\Gamma_1\rangle$. This concludes the proof.
\end{proof}

\ifshowproofs{
The morphing algorithm for binary trees we just presented has a simple extension to trees with unbounded degree. Namely, let $\Gamma_0$ and $\Gamma_1$ be any two order-preserving strictly-upward straight-line planar grid drawings of an $n$-node rooted ordered tree $T$. For $i=0,1$, we define an order-preserving strictly-upward straight-line planar grid drawing $\Gamma'_i$ of $T$ as follows; refer to \cref{fig:tree-upward-construction}. Let $x_{\Gamma'_i}(r(T))=0$ and let $y_{\Gamma'_i}(r(T))=y_{\Gamma_i}(r(T))$. Recursively construct a drawing of each subtree of $r(T)$. Then translate all such drawings horizontally so that: (i) the bounding boxes of the drawings of any two distinct subtrees of $r(T)$ can be separated by a vertical line; (ii) the left-to-right order of such bounding boxes corresponds to the left-to-right order of the children of $T$; and (iii) the edges from $r(T)$ to its children do not cross the nodes and the edges of the subtrees of $r(T)$. A proof similar to the one of \cref{th:binary-trees-upward-lineararea} shows that the morph $\langle \Gamma_0,\Gamma'_0,\Gamma'_1,\Gamma_1 \rangle$ is upward and planar.

\begin{figure}[tb]
    \centering
    \subfloat[\label{fig:tree-upward-construction}]
    {\includegraphics[height=.125\textwidth]{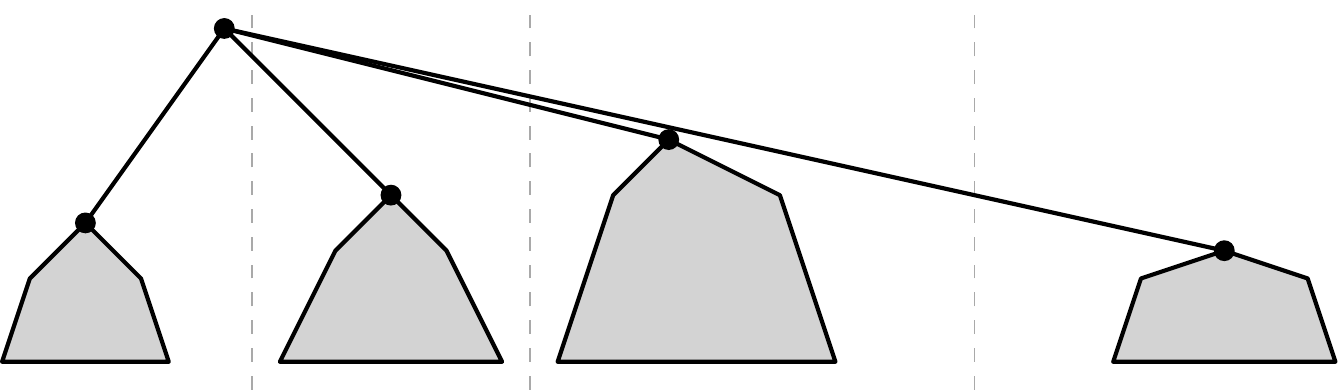}}
    \hfil
    \subfloat[]
    {\includegraphics[height=.125\textwidth]{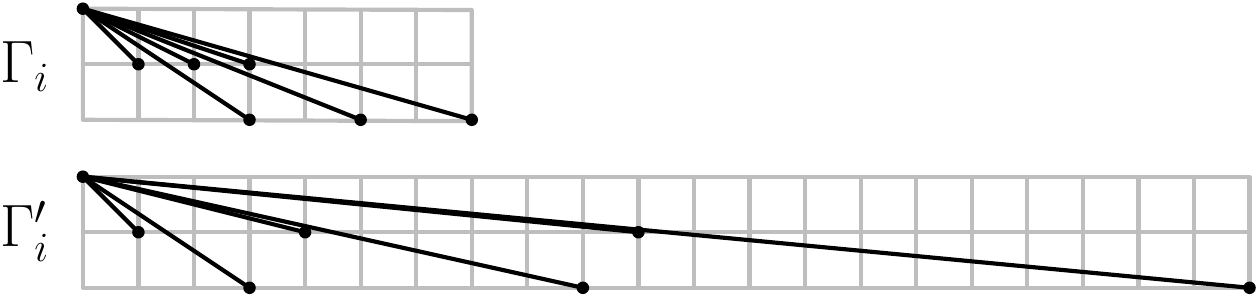}\label{fig:tree-upward-lowerbound}}
    \hfil
    \caption{(a) Construction of the drawing $\Gamma'_i$; the $y$-coordinate of each node is the same as in $\Gamma_i$. (b) A drawing $\Gamma_i$ such that the corresponding drawing $\Gamma'_i$ requires exponential width; in this illustration $h=2$.}
    \label{fig:tree-upward-constant-exponential}
  \end{figure}

Differently from the case of binary trees, however, the morph $\langle \Gamma_0,\Gamma'_0,\Gamma'_1,\Gamma_1 \rangle$ might require exponential area. Namely, consider a tree $T$ such that $r(T)$ has $n-1$ children. In their left-to-right order in $T$ such children alternate between \emph{top} and \emph{bottom children}. Let $\Gamma_i$ be the drawing of $T$ such that (see \cref{fig:tree-upward-lowerbound} top):

\begin{itemize}
	\item $r(T)$ has coordinates $(0,0)$;
	\item the top children of $r(T)$ have coordinates $(i,-1)$, for $i=1,2,\dots,\lfloor \frac{n}{2}\rfloor$; and 
	\item the bottom children of $r(T)$ have coordinates $(i\cdot h+1,-h)$, for $i=1,2,\dots,\lfloor\frac{n-1}{2}\rfloor$ and for some integer $h\geq 2$.
\end{itemize} 

Note that $h(\Gamma_i)\in O(h)$ and $w(\Gamma_i)\in O(n\cdot h)$, hence $\Gamma_i$ has polynomial area. Now consider any drawing $\Gamma'_i$ such that $y_{\Gamma'_i}(v)=y_{\Gamma_i}(v)$ for each node $v\in V(T)$ and such that properties (i)--(iii) above are satisfied; refer to \cref{fig:tree-upward-lowerbound} bottom. A linear number of children of $r(T)$ are on the same side of the vertical line $\ell$ through $r(T)$ in $\Gamma'_i$; say that $\Omega(n)$ children of $r(T)$ are to the right of $\ell$, the other case is analogous. Let $t_1,b_1,t_2,b_2,\dots,t_k$ be the left-to-right order of the children of $r(T)$ that are to the right of $\ell$ in $\Gamma'_i$, where the $t_i$'s are top children of $r(T)$ and the $b_i$'s are bottom children of $r(T)$; if the first child of $r(T)$ to the right of $\ell$ is a bottom child, we just disregard it. Suppose that the $x$-coordinate of $t_i$ is $x_i$; note that $x_1\geq 1$. Then the $x$-coordinate of $b_i$ is larger than $x_i \cdot h$, since the slope of $(r(T),t_i)$ is $-1/x_i$, since $(r(T),b_i)$ goes to the right of $(r(T),t_i)$, i.e., its slope is larger than $-1/x_i$ (by properties (ii) and (iii)), and since the vertical extension of $(r(T),b_i)$ is $h$. Since the $x$-coordinate of $t_{i+1}$ is larger than the one of $b_i$ (by properties (i) and (ii)), it follows that $x_{i+1}> x_i \cdot h$, which implies that $x_k > h^{k-1}$, hence $w(\Gamma'_i)\in h^{\Omega(n)}$.
}
\fi

\section{Planar Morphs of Tree Drawings}\label{se:general}

In this section we show how to construct small-area morphs between straight-line planar grid drawings of trees. In particular, we prove the following result.

\begin{theorem}\label{th:morph-straight-line-planar-drawings}
Let $T$ be an $n$-node ordered tree and let $\Gamma_0$ and $\Gamma_1$ be two order-preserving straight-line planar grid drawings of~$T$. There exists an $O(n)$-step planar morph $\mathcal M$ from $\Gamma_0$ to $\Gamma_1$ with $h(\mathcal M) \in O(D^3n \cdot H)$ and $w(\mathcal M) \in O(D^3 n \cdot W)$, where $H = \max\{h(\Gamma_0),h(\Gamma_1)\}$, $W = \max\{w(\Gamma_0),w(\Gamma_1)\}$, and $D = \max\{H,W\}$.
\end{theorem}

The rest of this section is devoted to the proof of \cref{th:morph-straight-line-planar-drawings}.
We are going to use the following definition (see \cref{fig:canonical}). 

\begin{figure}[tb] 
	\centering
	{\includegraphics[]{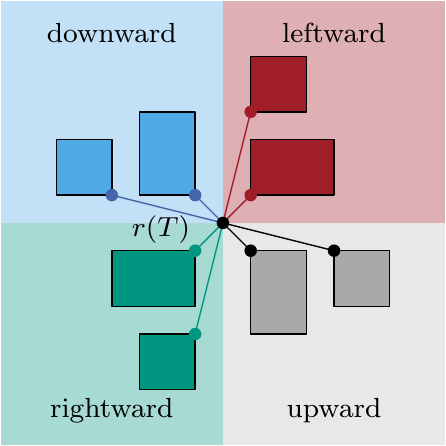}}
	\caption{Four canonical drawings of a tree $T$ (each shown in a differently colored quadrant).}
	\label{fig:canonical}
\end{figure}



\begin{definition}\label{def:canonical}
	An \emph{upward canonical drawing} of a rooted ordered tree $T$ is an order-preserving strictly-upward straight-line planar grid drawing $\Gamma$ of $T$ satisfying the following properties: 
	\begin{itemize}
		\item if $|V(T)|=1$, then $\Gamma$ is a grid point in the plane, representing $r(T)$;
		\item otherwise, let $\Gamma_1,\dots,\Gamma_k$ be upward canonical drawings of the subtrees $T_1,\dots,T_k$ of $r(T)$ (in their left-to-right order), respectively; then $\Gamma$ is such that:
		\begin{itemize}
			\item $r(T)$ is one unit to the left and one unit above the top-left corner of the bounding box of $\Gamma_1$;
			\item the top sides of the bounding boxes of $\Gamma_1, \dots, \Gamma_k$ have the same $y$-coordinate; and
			\item the right side of the bounding box of $\Gamma_i$ is one unit to the left of the left side of the bounding box of $\Gamma_{i+1}$, for $i=1, \dots, k-1$.		
		\end{itemize}
	\end{itemize}
\end{definition}

By counter-clockwise rotating an upward canonical drawing of $T$ by $\frac{\pi}{2}$, $\pi$, and $\frac{3\pi}{2}$ radians, we obtain a \emph{leftward}, a \emph{downward}, and a \emph{rightward canonical drawing} of $T$, respectively. A \emph{canonical drawing} of $T$ is an upward, leftward, downward, or rightward canonical drawing of $T$. Note that, in an upward, leftward, downward, or rightward canonical drawing $\Gamma$ of $T$, $r(T)$ is placed at the top-left, bottom-left, bottom-right, and top-right corner of the bounding box of $\Gamma$, respectively. 

\begin{remark}\label{remark:canonical-area}
If $T$ has $n$ nodes, then a canonical drawing of $T$ lies in the $2n$-box centered at $r(T)$.
\end{remark}

\begin{proof}
We prove the statement for an upward canonical drawing $\Gamma$; the proof for the other types of canonical drawings is analogous. First, we have $w(\Gamma)=n$, since each node has a distinct $x$-coordinate and, for each node $u$ of $T$, either there exists a node $v$ with $x(v)=x(u)+1$, or no node $v$ exists with $x(v)>x(u)$. Further, $h(\Gamma) \leq n$, since each node is vertically spaced from $r(T)$ by its graph-theoretic distance from it.
\end{proof}

The following lemma allows us to morph one canonical drawing into another in a constant number of morphing steps.

\begin{lemma}[Pinwheel]\label{lemma:rotate}
	Let $\Gamma$ and $\Gamma'$ be two canonical drawings of a rooted ordered tree $T$, where $r(T)$ is at the same point in $\Gamma$ and $\Gamma'$. If $\Gamma$ and $\Gamma'$ are
	\begin{inparaenum}[(i)]
		\item upward and leftward, or
		\item leftward and downward, or
		\item downward and rightward, or
		\item rightward and upward, respectively,
	\end{inparaenum}  
	then the morph $\langle \Gamma, \Gamma' \rangle$ is planar and lies in the interior of the right, top, left, or bottom half of the $2n$-box centered at $r(T)$, respectively.
\end{lemma}


\begin{proof}
We prove the statement in the case in which $\Gamma$ is upward and $\Gamma'$ is leftward; the other cases are symmetric. Without loss of generality, up to translation of the Cartesian axes, we assume that $r(T)$ lies at the point $o = (0,0)$.

We insert in $\Gamma$ and in $\Gamma'$ the drawings $B$ and $B'$, respectively, of a $4$-cycle $(r(T), a, b, c)$ as follows. 
The coordinates of $a$, $b$, and $c$ in $B$ are $a(0) = (0,-n)$, $b(0) = (n,-n)$, and $c(0) = (n,0)$, respectively.
The coordinates of $a$, $b$, and $c$ in $B'$ are $a(1) = (n,0)$, $b(1) = (n,n)$, and $c(1) = (0,n)$, respectively.
Note that the drawing of $T$ in $\Gamma$ ($\Gamma'$) lies strictly inside $B$ ($B'$), except for $r(T)$. For each node $u$ of $T$ denote by $\omega_u$, $\alpha_u$, $\beta_u$, and $\gamma_u$ non-negative reals such that $u(0) =\omega_u \cdot o + \alpha_u \cdot a(0) + \beta_u \cdot b(0) + \gamma_u \cdot c(0)$, where $u(0)$ is the position of $u$ in $\Gamma$ and $\omega_u + \alpha_u + \beta_u + \gamma_u = 1$. 

Consider the morph $\langle B,B'\rangle$ and assume it happens within the time interval $[0,1]$. Let $B(t)$ be the drawing of the cycle $(r(T), a, b, c)$ at time $t \in [0,1]$, where 
$B(0)=B$ and $B(1)=B'$.
Note that the coordinates of $r(T)$, $a$, $b$, and $c$ in $B(t)$ are $o = (0,0)$, $a(t)=(tn,(t-1)n)$, $b(t)=(n,(2t-1)n)$, and $c(t)=((1-t)n,tn)$, respectively.
For any $t \in [0,1]$, the affine transformation $\psi_t(\vec{x}):= \boldsymbol{A_t}\cdot\vec{x}$ with
$\boldsymbol{A_t} = 
\begin{pmatrix}
    1-t && -t \\
    t && 1-t
\end{pmatrix}
$ turns $B$ into $B(t)$.
%
%
Let $\Gamma(t)$ be the drawing of $T$ obtained by applying $\psi_t$ to $\Gamma$.
Since $\psi_t$ is an affine transformation and $det(\boldsymbol{A_t}) = (1-t)^2+t^2 \neq 0$, for any real value of $t$, we get that:
\begin{inparaenum}
\item for each node $u$ of $T$, the position $u(t)$ of $u$ in $\Gamma(t)$ is $u(t) =\omega_u \cdot o + \alpha_u \cdot a(t) + \beta_u \cdot b(t) + \gamma_u \cdot c(t)$, that is, the coefficients of the convex combination expressing the placement of $u$ in $\Gamma(t)$ with respect to the placement of $r(T)$, $a$, $b$, and $c$ are the same as in $\Gamma$; and
\item $\Gamma(t)$ is an order-preserving straight-line planar drawing of $T$.
\end{inparaenum}

It remains to prove that $\Gamma(t)$ is the drawing at time $t$ of the morph $\langle \Gamma, \Gamma' \rangle$, for any $t \in [0,1]$.
First, $\Gamma = \Gamma(0)$ since $\boldsymbol{A_0}$ is the identity matrix.
Second, since a leftward canonical drawing of $T$ is obtained by counter-clockwise rotating an upward canonical drawing of $T$ by $\frac{\pi}{2}$ and since $\boldsymbol{A_1}= 
\begin{pmatrix}
    0 && -1 \\
    1 && 0
\end{pmatrix}$
defines the same rotation, we have that $\Gamma'=\Gamma(1)$.
Finally, for any $0<t<1$ and for each node $u$ of $T$, the position of $u$ at time $t$ of the morph $\langle \Gamma, \Gamma' \rangle$ is 
{
\begin{eqnarray*}
&& (1-t)[\omega_u \cdot o + \alpha_u \cdot a(0) + \beta_u \cdot b(0) + \gamma_u \cdot c(0)] +t[\omega_u \cdot o + \alpha_u \cdot a(1) + \beta_u \cdot b(1) + \gamma_u \cdot c(1)] =\\
& = &\omega_u \cdot o + \alpha_u [a(0)\cdot (1-t) + a(1)\cdot t] + \beta_u [b(0)\cdot(1-t) + b(1)\cdot t] + \gamma_u [c(0)\cdot(1-t) + c(1)\cdot t] = \\
& = &\omega_u \cdot o + \alpha_u \cdot a(t) + \beta_u \cdot b(t) + \gamma_u \cdot c(t) = u(t).
\end{eqnarray*}
}
This concludes the proof that $\langle \Gamma, \Gamma' \rangle$ is planar.

Observe that, for each node $u$ of $T$, both $u(0)$ and $u(1)$ have their $x$-coordinates in the interval $[0,n-1]$ and their $y$-coordinates in the interval $[-n+1,n-1]$. This implies that, for any $t \in [0,1]$, the same holds for the coordinates of $u(t)$.
Therefore, the morph $\langle \Gamma, \Gamma' \rangle$ lies in the right half of the $2n$-box centered at $r(T)$.
\end{proof}

We now describe a proof of Theorem~\ref{th:morph-straight-line-planar-drawings}. Let $T$ be an $n$-node ordered tree and let $\Gamma_0$ and $\Gamma_1$ be two order-preserving straight-line planar grid drawings of~$T$. In order to compute a morph $\cal{M}$ from $\Gamma_0$ to $\Gamma_1$, we root $T$ at any leaf $r(T)$. Since $T$ is ordered, this determines a left-to-right order \mbox{of the children of each node.}



We are going to construct three morphs: a morph $\mathcal M^0$ from $\Gamma_0$ to a canonical drawing $\Gamma_0^*$ of $T$, a morph $\mathcal M^1$ from $\Gamma_1$ to a canonical drawing $\Gamma^*_1$ of $T$, and a morph $\mathcal M^{0,1}$ from $\Gamma_0^*$ to $\Gamma^*_1$. The morph $\mathcal M$ is then obtained by composing $\mathcal M^0$, $\mathcal M^{0,1}$, and the reverse of $\mathcal M^1$. 
The morph $\mathcal M^{0,1}$ consists of $O(1)$ steps and can be constructed by applying \cref{lemma:rotate}.
We describe below how to construct $\mathcal M^0$; the construction of $\mathcal{M}^1$ is analogous. However, before describing the construction of $\mathcal M^0$, we introduce a labeling of the nodes of $T$ and the concept of ``partially-canonical drawing''.

Let $T[0]$ be the tree $T$ together with a labeling of each of the $k$ internal nodes of $T$ as \texttt{unvisited} and of each leaf as \texttt{visited}. We perform a bottom-up visit of $T$, labeling one-by-one the internal nodes of $T$ as \texttt{visited}. We label a node $v$ as \texttt{visited} only after all of its children have been labeled as \texttt{visited}. For $i=0,\dots,k$, we denote by $T[i]$ the tree $T$ once $i$ of its internal nodes have been labeled as \texttt{visited}. 

The outline of our algorithm for constructing $\mathcal M^0$ is as follows. In a first morphing step, we scale $\Gamma_0$ up in order to make some free room around each node. Then we process the nodes of $T$ and label them as \texttt{visited} one by one, as described above. When we label a node $v$ as \texttt{visited}, we morph the current drawing into one in which $T_v$ is upward or downward canonical; this is accomplished by only moving the subtrees rooted at the children of $v$. Note that, when $v$ is labeled as \texttt{visited}, all the children of $v$ are already labeled as \texttt{visited}, hence the drawings of the subtrees rooted at them are upward or downward canonical. Thus, the movement of such drawings only consists of translations and rotations to bring such drawings where they need to be in the upward or downward canonical drawing of $T_v$. The initial scaling ensures that there is enough room around $v$ so that an upward or downward canonical of $T_v$ does not intersect the rest of the drawing. In order to formalize this process, we need to describe the properties of the drawing that we obtain after a number of nodes of $T$ have been labeled as \texttt{visited}; we call \emph{partially-canonical} such a drawing.

%
%

Let $D_0 = \max\{w(\Gamma_0),h(\Gamma_0)\}$. 
Let $\Gamma$ be a drawing of $T$ and let $v$ be a node of $T$. We denote by 
$\largebox{v}$, $\mediumbox{v}$, and $\smallbox{v}$ the 
$(\ell_0+4n)$-box, the $\ell_0$-box, and the $2n$-box centered at $v$ in $\Gamma$, respectively, where $\ell_0 = k_0 D_0^2n$ for some constant $k_0 > 1$ to be determined later. We have the following definition.

\begin{figure}[tb] 
	\centering
	\includegraphics[width=.5\textwidth,page=1]{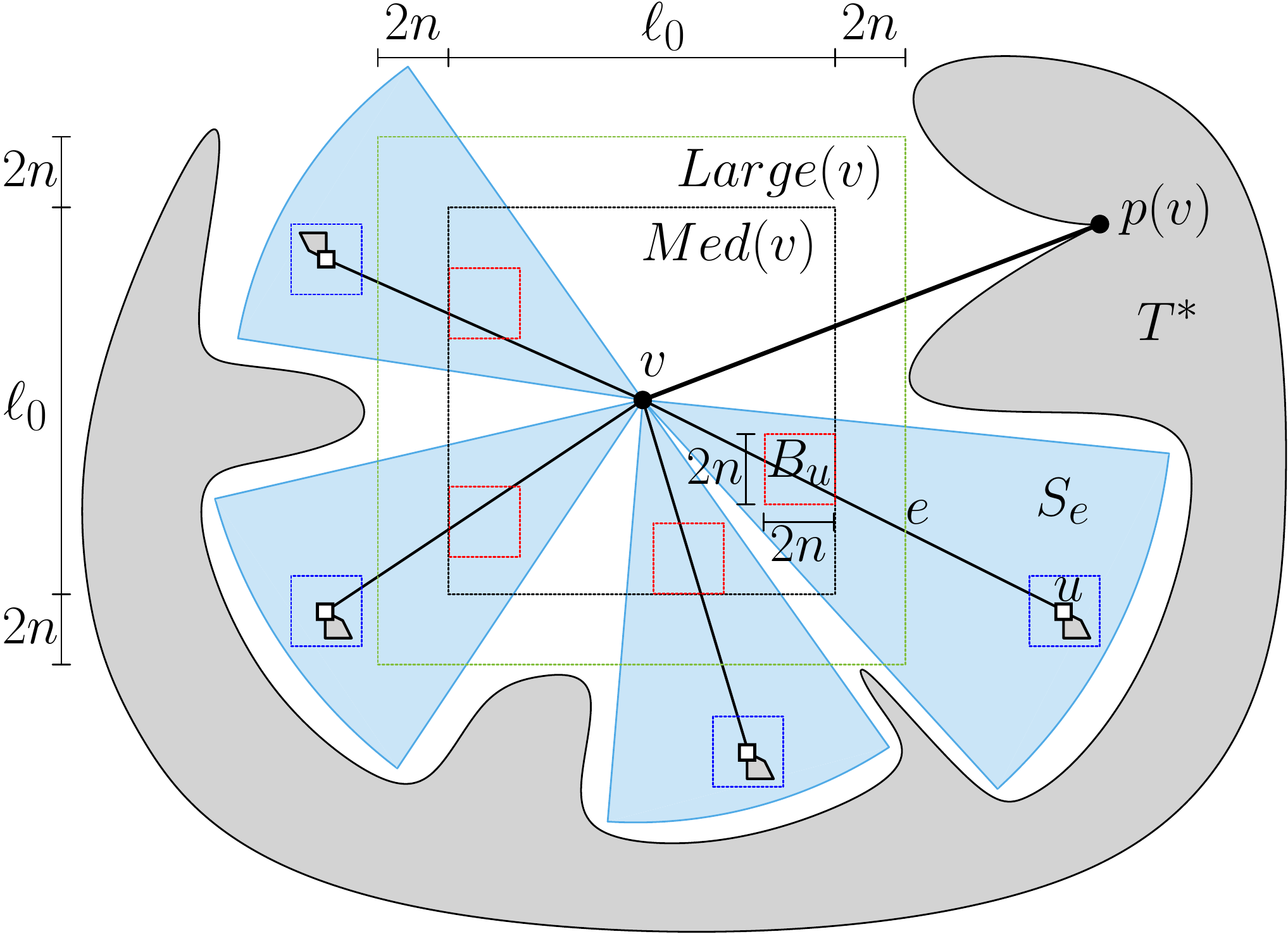}
	\caption{A partially-canonical drawing $\Gamma$ of $T[i]$. The illustration focuses the attention on an \texttt{unvisited} $v$ whose children are all \texttt{visited}.}
	\label{fig:partially-canonical}
\end{figure}

\begin{definition}\label{def:partially-canonical}
	An order-preserving straight-line planar grid drawing $\Gamma$ of $T$ is a \emph{partially-canonical drawing} of $T[i]$ if it satisfies the following properties (refer to \cref{fig:partially-canonical}):
	\begin{enumerate}[(a)]
		\item \label{partially-canonical-a} for each \texttt{visited} node $u$ of $T$, the drawing $\Gamma_u$ of $T_u$ in $\Gamma$ is upward canonical or downward canonical; further, if $u\neq r(T)$, then $\Gamma_u$ is upward canonical, if $y_\Gamma(u) \le y_\Gamma(p(u))$, or downward canonical, if $y_\Gamma(u) > y_\Gamma(p(u))$;
		\item \label{partially-canonical-c} for each edge $e=(v,u)$ of $T$, where $v$ is the parent of $u$ and $v$ is \texttt{unvisited}, there exists a sector $S_e$ of a circumference centered at $v$ such that:
		\begin{enumerate}[(b.i)]
			\item $S_e$ encloses $\smallbox{u}$; 
			\item $S_e$ contains no node with the exception of $v$ and of, possibly, the nodes of $T_u$, and no edge with the exception of $e$ and of, possibly, the edges of $T_u$;
			\item the intersection between $S_e$ and $\mediumbox{v}$ contains a $2n$-box $B_u$ whose corners have integer coordinates and whose center $c_u$ is such that $y_\Gamma(c_u) \leq y_\Gamma(v)$ if and only if $y_\Gamma(u) \leq y_\Gamma(v)$; 
			and 
			\item for any edge $e' \neq e$ incident to $v$, the sectors $S_e$ and $S_{e'}$ are internally disjoint;
		\end{enumerate}
	\item for any two \texttt{unvisited} nodes $v$ and $w$, it holds $\largebox{v} \cap \largebox{w} = \emptyset$; and
	\item \label{partially-canonical-b} for any \texttt{unvisited} node $v$ of $T$, $\largebox{v}$ contains no node different from $v$, and any edge $e$ or any sector $S_e$ intersecting $\largebox{v}$ is such that $e$ is incident to $v$.
	\end{enumerate}
\end{definition}

Note that, by Property~(a), a partially-canonical drawing of $T[k]$ is a canonical drawing of $T$.

The algorithm to construct $\mathcal M^0$ is as follows. First, we scale $\Gamma_0$ up by a factor in $O(D_0^3n)$ so that the resulting drawing $\Delta_0$ is a partially-canonical drawing of $T[0]$ (see \cref{lemma:scaling}). Clearly, the morph $\mathcal M_0 = \langle \Gamma_0, \Delta_0 \rangle$ is planar, $w(\mathcal M_0) = w(\Delta_0)$, and $h(\mathcal M_0) = h(\Delta_0)$.

For $i=1,\dots, k$, let $v_i$ be the node that is labeled as \texttt{visited} at the $i$-th step of the bottom-up visit~of~$T$.
Starting from a partially-canonical drawing $\Delta_{i-1}$ of $T[i-1]$, we construct a partially-canonical drawing $\Delta_{i}$ of $T[i]$ and a morph $\mathcal M_{i-1,i}$ from $\Delta_{i-1}$ to $\Delta_{i}$ with $O(\deg(v_i))$ steps, with $w(\mathcal M_{i-1,i}) \leq  w(\Delta_0) + \ell_0 + 4n$ and $h(\mathcal M_{i-1,i}) \leq h(\Delta_0) + \ell_0 + 4n$ (see \cref{lemma:partially-canonical-from-partially-canonical}).

Composing the morphs $\mathcal M_0, \mathcal M_{0,1},\mathcal M_{1,2},\dots,\mathcal M_{k-1,k}$ yields the desired morph $\mathcal M^0$ from~$\Gamma_0$ to a canonical drawing $\Delta_k=\Gamma^*_0$ of $T$. The morph has $O(\sum_{i=1}^k deg(v_i))\subseteq O(n)$ steps (by \cref{lemma:partially-canonical-from-partially-canonical}). Further, $w(\mathcal M^0) \leq  w(\Delta_0) + \ell_0 + 4n$ and $h(\mathcal M^0) \leq h(\Delta_0) + \ell_0 + 4n$ (again by \cref{lemma:partially-canonical-from-partially-canonical}), hence $w(\mathcal M^0) \in O(D^3_0n \cdot w(\Gamma_0))$ and $h(\mathcal M^0) \in O(D^3_0n \cdot h(\Gamma_0))$ (by \cref{lemma:scaling}).

It remains to show how to construct the partially-canonical drawing $\Delta_i$, for $i=0,\dots,k$, and the morphs $\mathcal M_{i-1,i}$, for $i=1,\dots,k$. Specifically, in the following we prove the next two lemmas.

\begin{lemma}\label{lemma:scaling}
	There is an integer $B_0 \in O(D^3_0n)$ such that the drawing $\Delta_0$ obtained by scaling the drawing $\Gamma_0$ of $T$ up by $B_0$ is a partially-canonical drawing~of~$T[0]$.
\end{lemma}

\begin{lemma}\label{lemma:partially-canonical-from-partially-canonical}
For any $i\in \{1,\dots,k\}$, let $\Delta_{i-1}$ be a partially-canonical drawing of $T[i-1]$. There exists a partially-canonical drawing $\Delta_i$ of $T[i]$ and an $O(\deg(v_i))$-step planar morph $\mathcal M_{i-1,i}$ from $\Delta_{i-1}$ to $\Delta_{i}$ such that $w(\mathcal M_{i-1,i}) \leq w(\Delta_0) + \ell_0 + 4n$ and $h(\mathcal M_{i-1,i}) \leq  h(\Delta_0) + \ell_0 + 4n$.
\end{lemma}

\ifshowproofs{
	\subsection{Proof of \cref{lemma:scaling}}\label{sse:lemma-scaling}
	
In order to prove \cref{lemma:scaling}, we show that there exists a constant $\beta_0$ such that setting $B_0 = \beta_0 D_0^3 n$ guarantees that Properties (a)--(d) of \cref{def:partially-canonical} are satisfied by the drawing $\Delta_0$ of $T[0]$ obtained by scaling the drawing $\Gamma_0$ of $T$ up by $B_0$. In the following, we often implicitly exploit $n\geq 1$ and $D_0\geq 1$.

	Regarding Property (a), recall that a node $u$ of $T$ is labeled \texttt{visited} in $T[0]$ if and only if $u$ is a leaf. Hence, for a \texttt{visited} node $u$ of $T[0]$, we have that $T_u=u$ and the drawing of $T_u$ in $\Delta_0$ is both upward canonical and downward canonical, trivially satisfying Property (a) of \cref{def:partially-canonical}.
	
	\begin{figure}[tb] 
		\centering
		\hfill
		\subfloat[\label{fig:prop-i}]
		{\includegraphics[scale=.75,page=2]{props.pdf}}
		\hfill
		\subfloat[\label{fig:prop-iii}]
		{\includegraphics[scale=.75,page=1]{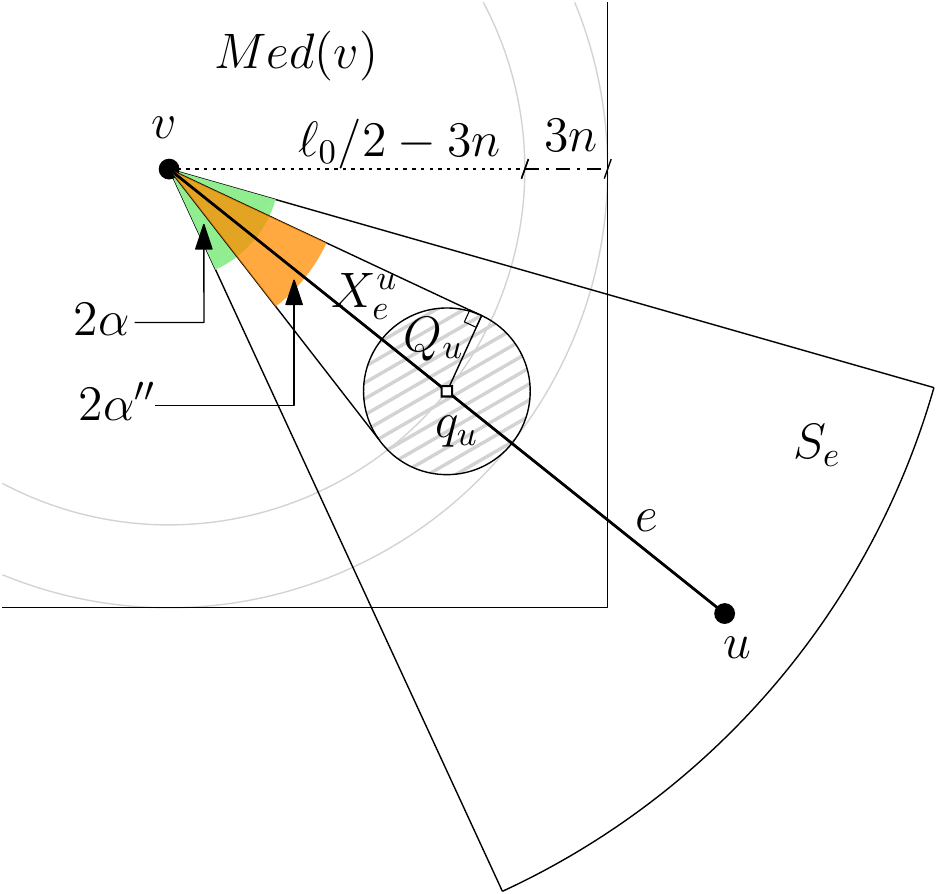}}
		\hfill
		\\
		\hfill
		\subfloat[\label{fig:prop-iv}]
		{\includegraphics[scale=.75,page=3]{props.pdf}}
		\hfill
		\subfloat[\label{fig:prop-ii}]
		{\includegraphics[scale=.75,page=4]{props.pdf}}
		\hfill
		\caption{Illustrations for the proof of \cref{lemma:scaling}. (a) Property~\emph{(b.i)}. (b) Property~\emph{(b.iii)}. (c) Property~\emph{(b.iv)}. (d) Property~\emph{(b.ii)}.}
	\end{figure}
	
	We now show that Properties (b.i)--(b.iv) of \cref{def:partially-canonical} hold.
	We denote by $\ell_{\Gamma}(e)$ the length of an edge $e$ in a drawing $\Gamma$.
	For every edge $e=(v,u)$ of $T$ where $v$ is the parent of $u$ and $v$ is \texttt{unvisited}, let
	$C_e$ be a circumference centered at $v$ whose radius is a value $r_e = \ell_{\Delta_0}(e) + \sqrt{2}n$; also,
	let $W_e$ be a wedge whose central angle is bisected by $e$ and has a value $2\alpha$ to be determined later.
	Let $S_e$ be the sector of $C_e$ determined by the intersection of $C_e$ and $W_e$.
	We remark that $r_e$ depends on the length of the edge $e$, whereas $\alpha$ is the same for all the sectors.

	We start by showing that Property (b.i) holds; refer to \cref{fig:prop-i}. In fact, we prove a stronger statement, namely that $S_e$ encloses the disk $K_u$ centered at $u$ with radius $\sqrt{2}n$; note that $K_u$ encloses $\smallbox{u}$. 
	Let $W^u_e$ be the wedge containing $u$, centered at $v$, and delimited by the two lines passing through $v$ that are tangent to $K_u$. Denote by $2\alpha'$ the angle spanned by $W^u_e$.
	In order for $S_e$ to enclose $K_u$, it must happen that both $C_e$ and $W_e$ enclose $K_u$. Actually, the definition of $r_e$ directly implies that $C_e$ encloses $K_u$. Thus, in order to prove that $W_e$ encloses $K_u$, we only need to show that $\alpha' \leq \alpha$.
	By looking at the right triangle whose corners are $u$, $v$, and one of the intersection points of the lines delimiting $W^u_e$ with the boundary of $K_u$, we have that $\alpha' = \arcsin\left(\frac{\sqrt{2}n}{\ell_{\Delta_0}(e)}\right)$. Moreover, since $\ell_{\Gamma_0}(e) \geq 1$ and since $B_0 = \beta_0 D_0^3 n$, we have $\ell_{\Delta_0}(e) \geq \beta_0 D_0^3 n$. Therefore, we have $\alpha' \leq  \arcsin\left(\frac{\sqrt{2}n}{\beta_0 D_0^3 n}\right)$ and Property (b.i) is satisfied as long as
	
	\begin{equation}\label{cr:iv}
	\alpha \geq \arcsin\left(\frac{\sqrt{2}}{\beta_0 D_0^3}\right).
	\end{equation}
	
	Next, we show that Property (b.iii) holds; refer to \cref{fig:prop-iii}. Consider the disk $Q_u$ with radius $3n$ centered at a point $q_u$ of the straight-line segment representing $e$ in $\Delta_0$ at distance $\frac{\ell_0}{2}-3n$ from $v$. 
	
	We first prove that $Q_u$ contains a $2n$-box $B_u$ satisfying the conditions of Property (b.iii).
	Note that $Q_u$ lies inside $\mediumbox{v}$. Further, $Q_u$ does not contain $v$ if $\frac{\ell_0}{2} - 3n > 3n$, that is $\frac{k_0 D_0^2 n}{2} > 6n$, which is true as long as 
	
	\begin{equation}\label{eq:k0}
	k_0 > 12.
	\end{equation}
		
	If $q_u$ is a grid point, then the $2n$-box $B_u$ whose center is $c_u = q_u$ lies inside $\mediumbox{v}$ (since it lies inside $Q_u$), has corners with integer coordinates, and it is such that $y_{\Delta_0}(c_u) \leq y_{\Delta_0}(v)$ if and only if $y_{\Delta_0}(u) \leq y_{\Delta_0}(v)$.
	If $q_u$ is not a grid point, then consider the grid cell containing $q_u$. Let $c_u$ be a corner of this cell such that
	$y_{\Delta_0}(c_u) \leq y_{\Delta_0}(v)$ if and only if $y_{\Delta_0}(u) \leq y_{\Delta_0}(v)$. Note that the distance between $q_u$ and $c_u$ is less than $\sqrt{2}$. Therefore, the $2n$-box $B_u$ whose center is $c_u$ lies inside $Q_u$ (given that $3n > \sqrt{2}n + \sqrt{2}$) and has corners with integer coordinates.
	
	We now prove that $S_e$ encloses $Q_u$.
	Let $X^u_e$ be the wedge containing $q_u$, centered at $v$, and delimited by the two lines passing through $v$ that are tangent to $Q_u$. Denote by $2\alpha''$ the angle spanned by $X^u_e$.
	In order for $S_e$ to enclose $Q_u$, it must happen that $W_e$ encloses $Q_u$; this is implied by $\alpha'' \leq \alpha$, which we ensure next.
	By looking at the right triangle whose corners are $q_u$, $v$, and one of the intersection points of the lines delimiting $X^u_e$ with the boundary of $Q_u$, we have that $\alpha'' = 
	\arcsin\left(\frac{3n}{\ell_0/2 - 3n}\right) = \arcsin\left(\frac{3n}{k_0 D_0^2n/2 - 3n}\right)$. By~\cref{eq:k0}, we have $k_0 > 12$, which implies that $3n < \frac{k_0 D_0^2 n}{4}$. Hence, we have that $\alpha'' \leq 
	\arcsin\left(\frac{3n}{k_0 D_0^2n/2 - k_0 D_0^2n/4}\right)=\arcsin\left(\frac{12}{k_0 D_0^2}\right)$, which implies that Property (b.iii) is satisfied if 
	
	\begin{equation}\label{cr:v}
	\alpha \geq \arcsin\left(\frac{12}{k_0 D_0^2}\right).
	\end{equation}
	
	Before discussing Property (b.iv), we prove the following.
		
		\begin{claim}\label{claim:minimum-grid-distance}
			Let $A=(x_A,y_A)$, $B=(x_B,y_B)$, and $C=(x_C,y_C)$ be three non-collinear points on a $D \times D$ grid. Then the distance between $A$ and the line containing the segment $\overline{BC}$ is at least $\frac{1}{\sqrt{2}D}$.
		\end{claim}
		
		\begin{proof}
			The distance in the statement is equal to 
			\begin{equation}\label{eq:distance}
			\frac{|(y_C-y_B)x_A - (x_C-x_B)y_A + x_C y_B - x_B y_C|}{\sqrt{ (y_C -y_B)^2 + (x_C -x_B)^2 }}.
			\end{equation}

			Since the numerator of \cref{eq:distance} is the modulus of an expression only containing multiplications and sums between integers, its value is a non-negative integer. Also, since points $A$, $B$, and $C$ are not collinear, its value is at least $1$.
			The square root at the denominator of \cref{eq:distance} contains the sum of two values, each of which is at most $D^2$.
			Hence, the value of the denominator of \cref{eq:distance} is at most $\sqrt{2} D$.
			We thus have that the distance in the statement is at least $\frac{1}{\sqrt{2} D}$.
		\end{proof}

	Next, we show that Property (b.iv) holds; refer to \cref{fig:prop-iv}.
	Let $2\beta$ be the smallest angle formed by any two edges $e$ and $e'$ incident to $v$. In order to prove that $S_e$ and $S_{e'}$ are internally disjoint, it suffices to ensure that $\alpha < \beta$. 
	In the case in which $\beta \geq 45^\circ$, we ensure that $\alpha < \beta$ by the constraint
	
	\begin{equation}\label{cr:vi}
	\alpha < 45^\circ.
	\end{equation}

	Suppose that $\beta < 45^\circ$. Without loss of generality, let $e=(v,w)$ be the shortest between $e$ and $e'$ in $\Gamma_0$. Let $p$ be the projection of $w$ onto $e'$. 
	We have that $\ell_{\Gamma_0}(e) \leq D_0$ and that the length $||\overline{wp}||$ of the segment connecting $w$ and $p$ in $\Gamma_0$ is at least $\frac{1}{\sqrt{2}D_0}$ by \cref{claim:minimum-grid-distance}.
	By looking at the right triangle whose corners are $v$, $w$, and $p$,  we have that
	$\sin(2\beta) =  \frac{||\overline{wp}||}{\ell_{\Gamma_0}(e)} \geq \frac{1}{\sqrt{2}D_0^2}$.
	Hence, by \cref{cr:vi} we have that $\alpha < \beta$ as long as $\sin(2\alpha) \leq 2 \sin(\alpha) < \sin(2 \beta)$, which is ensured by the following constraint
	
	\begin{equation}\label{cr:vii}
	\alpha < \arcsin\left(\frac{1}{2 \sqrt{2} D_0^2}\right).
	\end{equation}
	
	We now prove that Property (b.ii) holds; refer to \cref{fig:prop-ii}. We are going to use the following.

	\begin{claim}\label{claim:circumference}
		Let $\mathcal C$ be the closed disk defined by the inequality $x^2+y^2\leq 1$. Let $x_1\geq x_0>0$ and let $\mathcal C_1$ be the subset of $\mathcal C$ whose points have $x$-coordinate greater than or equal to $x_1$. The points of $\mathcal C_1$ that are farthest from $(x_0,0)$ are $(x_1,\pm \sqrt{1-x_1^2})$.
	\end{claim}
		
	\begin{proof}
		First, note that the maximum distance from $(x_0,0)$ is achieved by a point on the boundary of $\mathcal C_1$. Let $(x,y)$ be any point on the boundary of $\mathcal C_1$. The distance between $(x_0,0)$ and $(x,y)$ is $d(x,y)=\sqrt{x^2-2xx_0 +x_0^2+y^2}$. Since $x^2+y^2=1$, we have $d(x,y)=\sqrt{1-2xx_0 +x_0^2}$. The derivative of $d(x,y)$ with respect to $x$ is equal to $\frac{-x_0}{\sqrt{1-2xx_0 +x_0^2}}$, which is negative when $x_0>0$. Hence, the maximum of $d(x,y)$ is achieved when $x=x_1$ and $y=\pm \sqrt{1-x_1^2}$.
	\end{proof}

	We prove that the distance in $\Delta_0$ between any point of the sector $S_e$ and $e$ is smaller than the distance between $e$ and any node or edge of $T-\{T_u + e\}$; that is, all the nodes and edges of $T$ that are not in $T_u + e$ are entirely outside $S_e$. Note that the distance between $e$ and any node or edge of $T-\{T_u + e\}$ is at least $\frac{\beta_0 D_0^2n}{\sqrt{2}}$, given that it is at least $\frac{1}{\sqrt{2}D_0}$ in $\Gamma_0$ by~\cref{claim:minimum-grid-distance}. 	
	
	Denote by $\kappa_1$ and $\kappa_2$ the intersections between the rays delimiting $W_e$ and $C_e$. Consider the line $\ell_u$ perpendicular to $e$ passing through $u$. We distinguish two cases, based on whether $\ell_u$ intersects the boundary of $S_e$  in $C_e$ or on the segments $\overline{v\kappa_1}$ and $\overline{v\kappa_2}$. 
	
	\begin{itemize}
		\item In the former case, consider the intersection point $p_{\kappa}$ between the segment $\overline{\kappa_1\kappa_2}$ and the edge $e$. Let $\mathcal K$ be the region of the plane whose boundary is the triangle with corners $v$, $\kappa_1$, and $\kappa_2$. The distance between any point of $\mathcal K$ and $e$ is smaller than or equal to $||\overline{p_{\kappa}\kappa_1}||$. Further, by~\cref{claim:circumference}, the distance between any point of $S_e\setminus \mathcal K$ and $p_{\kappa}$ is also at most $||\overline{p_{\kappa}\kappa_1}||$. Note that $||\overline{p_{\kappa}\kappa_1}||\leq \ell_{\Delta_0}(e) \tan(\alpha)$.
		\item In the latter case, consider the intersection points $\lambda_1$ and $\lambda_2$ between the rays delimiting $W_e$ and $\ell_u$. Let $\Lambda$ be the region of the plane whose boundary is the triangle with corners $v$, $\lambda_1$, and $\lambda_2$; further, let $\Lambda_1$ ($\Lambda_2$) be the region of the plane whose boundary is the triangle with corners $v$, $\lambda_1$, and $\kappa_1$ (resp.\ $v$, $\lambda_2$, and $\kappa_2$); finally, let $\Lambda_C$ be the region of the plane delimited by the segments $\overline{u\kappa_1}$ and $\overline{u\kappa_2}$, and by the arc of $C_e$ between $\kappa_1$ and $\kappa_2$. Note that $S_e = \Lambda \cup \Lambda_1 \cup \Lambda_2 \cup \Lambda_C$. 
		
		The distance between any point of $\Lambda$ and $e$ is smaller than or equal to $||\overline{u\lambda_1}||$. Further, the distance between any point of $\Lambda_1$ ($\Lambda_2$) and $e$ is smaller than or equal to $||\overline{u\kappa_1}||$ (resp.\ $||\overline{u\kappa_2}||$). Finally, by~\cref{claim:circumference}, the distance between any point of $\Lambda_C$ and $e$ is smaller than or equal to $||\overline{u\kappa_1}||$. Since $\lambda_1$ ($\lambda_2$) belongs to both $\Lambda$ and $\Lambda_1$ (resp.\ $\Lambda$ and $\Lambda_2$) and since $\kappa_1$ and $\kappa_2$ both belong to $\Lambda_C$, it follows that the distance between any point of $S_e$ and $e$ is smaller than or equal to~$||\overline{u\kappa_1}||$. 
		
		Since $||\overline{v \lambda_1}|| > ||\overline{v u}||$, we have $||\overline{\lambda_1 \kappa_1}|| < \sqrt{2}n$. Further, by the triangular inequality, we have that  $||\overline{u \kappa_1}|| < ||\overline{u \lambda_1}|| + ||\overline{\kappa_1 \lambda_1}||< \ell_{\Delta_0}(e) \tan(\alpha) + \sqrt{2}n$.
	\end{itemize}

	In order to prove that Property (b.ii) holds, we show that $\ell_{\Delta_0}(e) \tan(\alpha) + \sqrt{2}n < \frac{\beta_0 D_0^2n}{\sqrt{2}}$. Since $\tan(\alpha)=\frac{\sin(\alpha)}{\cos  (\alpha)}$ and by the constraint $\alpha<45^\circ$, we have  $\tan(\alpha)\leq \sqrt 2\sin(\alpha)$. Hence, $\ell_{\Delta_0}(e) \tan(\alpha) + \sqrt{2}n < \sqrt 2\sin(\alpha)\ell_{\Delta_0}(e) + \sqrt{2}n < \sqrt 2\sin(\alpha)\beta_0 D_0^4 n + \sqrt{2}n<2\sqrt 2\sin(\alpha)\beta_0 D_0^4 n$. Thus, we have $2\sqrt 2\sin(\alpha)\beta_0 D_0^4 n< \frac{\beta_0 D_0^2n}{\sqrt{2}}$ as long as $4 D_0^2 \sin(\alpha) < 1$. This gives us the constraint
	
	\begin{equation}\label{cr:viii}
	\alpha < \arcsin\left(\frac{1}{4 D^2_0}\right).
	\end{equation}

	Concerning Property (c), note that the minimum distance between two nodes in $\Gamma_0$ is $1$, given that $\Gamma_0$ is a grid drawing. Therefore, the minimum distance between two nodes in $\Delta_0$ is at least $\beta_0 D^3_0 n$.
	Also, for any \texttt{unvisited} node $v$ of $T[0]$, the points of $\largebox{v}$ that are furthest from $v$ are its corners. Such points are
	at distance $\sqrt{2} \frac{\ell_0 + 4n}{2} = \sqrt{2} \frac{k_0 D_0^2 n + 4n}{2} \leq \frac{\sqrt{2}}{2}(k_0+4)D_0^2n$ from $v$. Therefore, there exist no two \texttt{unvisited} nodes $v$ and $w$ of $T[0]$ whose corresponding boxes $\largebox{v}$ and $\largebox{w}$ intersect in $\Delta_0$, as long as $||\overline{vw}||\geq \beta_0 D^3_0 n> \sqrt{2}(k_0+4)D_0^2n$, which is true if the following holds
	
	\begin{equation}\label{cr:i}
	\beta_0 > \sqrt{2}(k_0+4).
	\end{equation}

	Consider Property (d) of \cref{def:partially-canonical}. 
	The fact that, for any \texttt{unvisited} node $v$, the box $\largebox{v}$ does not contain any node different from $v$ follows from the same arguments exploited for Property (c), as long as \cref{cr:i} holds.
	

	We prove that, for any edge $e$ that is not incident to $v$, the sector $S_e$ does not intersect $\largebox{v}$; since $S_e$ contains $e$, this implies that $e$ does not intersect $\largebox{v}$ either. Let $s$ be the straight-line segment representing $e$ in $\Delta_0$. As argued when proving Property~(b.ii), any point of $S_e$ is at distance at most $\ell_{\Delta_0}(e) \tan(\alpha) + \sqrt{2}n < 2\sqrt 2\sin(\alpha)\beta_0 D_0^4 n$ from $s$. 	Further, the maximum distance of any point of $\largebox{v}$ from $v$ is $\sqrt{2} \frac{k_0 D_0^2 n + 4n}{2}<\sqrt 2 k_0 D_0^2 n$, where we exploited \cref{eq:k0}. By \cref{claim:minimum-grid-distance}, we have that the minimum distance between $s$ and $v$ is at least $\frac{1}{\sqrt{2} D_0}$ in $\Gamma_0$ and, thus, at least $\frac{\beta_0}{\sqrt{2}} D_0^2 n$ in $\Delta_0$.
	Therefore, the box $\largebox{v}$ is not traversed by $S_e$ as long as
	$\frac{\beta_0}{\sqrt{2}} D_0^2 n  > \sqrt 2 k_0 D_0^2 n + 2\sqrt 2\sin(\alpha)\beta_0 D_0^4 n$, that is $\beta_0 > 2 k_0 + 4\sin(\alpha)\beta_0 D_0^2$. We have $4\sin(\alpha)\beta_0 D_0^2 < \beta_0/2$ as long as 
	
	\begin{equation}\label{cr:ix}
	\alpha < \arcsin\left(\frac{1}{8 D_0^2}\right),
	\end{equation}
	
	\noindent hence Property~(d) is satisfied as long as 
	
	\begin{equation}\label{cr:iii}
	\beta_0 > 4 k_0. 
	\end{equation}

	We now choose $k_0 = 150$, $\beta_0 = 800$, and $\alpha = \arcsin(\frac{0.1}{D_0^2})$. 
	We have that \cref{eq:k0,cr:i,cr:iii} are satisfied by the choice of $k_0$ and $\beta_0$.
	Inequalities \cref{cr:vi,cr:vii,cr:viii} are weaker than inequality \cref{cr:ix}, which is true since $\frac{1}{8} > 0.1$.
	Inequality \cref{cr:v} holds true since $\frac{12}{k_0}<0.1$.
	Inequality \cref{cr:iv} holds true since $\frac{\sqrt{2}}{\beta_0}<0.1$ and $D_0 \geq 1$.
}
\fi

\subsection{Proof of \cref{lemma:partially-canonical-from-partially-canonical}}\label{sse:lemma-new-canonical}

We denote by $T^*$ the tree obtained by removing $T_{v_i}$ from $T$.
Let $\Delta_i$ be the drawing of $T$ obtained from $\Delta_{i-1}$ by redrawing $T_{v_i}$ so that it is upward canonical, if $v_i=r(T)$ or if $v_i\neq r(T)$ and $y_{\Delta_{i-1}}(v_i) \leq y_{\Delta_{i-1}}(p(v_i))$, or downward canonical, otherwise, while keeping the placement of $v_i$ and of every node of $T^*$ unchanged. We have the following.

\begin{lemma}\label{lem:partially-canonical-delta-i}
The drawing $\Delta_i$ is a partially-canonical drawing of $T[i]$.
\end{lemma}

\begin{proof}
By construction, the drawing $\Delta_i$ is such that:
\begin{enumerate}[(1)]
	\item the drawing of $T^*$ is the same both in $\Delta_i$ and in $\Delta_{i-1}$, 
	\item $v_i$ is in the same position in $\Delta_i$ and in $\Delta_{i-1}$, and
	\item the drawing of $T_{v_i}$ in $\Delta_i$ is upward canonical, if $v_i=r(T)$ or if $v_i\neq r(T)$ and  $y_{\Delta_{i-1}}(v_i) \leq y_{\Delta_{i-1}}(p(v_i))$, or downward canonical, otherwise.
\end{enumerate}	
	
The drawing $\Delta_i$ is straight-line by construction and it is an order-preserving grid drawing since $\Delta_{i-1}$ is. The planarity of $\Delta_i$ can be proved as follows. First, the drawing of $T^*+(v_i,p(v_i))$ in $\Delta_i$ is planar since it coincides with the drawing of $T^*+(v_i,p(v_i))$ in $\Delta_{i-1}$, by~(1) and~(2), and since $\Delta_{i-1}$ is planar.  Second, the drawing of $T_{v_i}$ in $\Delta_i$ is planar since it is a canonical drawing, by construction. Third, the drawing of $T_{v_i}$ in $\Delta_i$ does not intersect any node or edge of $T^*$; namely, the drawing of $T_{v_i}$ in $\Delta_i$ is contained in $\smallbox{v_i}$, by~\cref{remark:canonical-area}; further $\smallbox{v_i}$ is contained in $\largebox{v_i}$, by definition; finally, $\largebox{v_i}$ contains no node or edge of $T^*$, by Property (d) of $\Delta_{i-1}$. Fourth, the drawing of $T_{v_i}$ in $\Delta_i$ does not intersect the edge $(v_i,p(v_i))$, if such an edge exists, by~(3).  
	
We now show that $\Delta_i$ satisfies Properties (a)--(d) of \cref{def:partially-canonical}.

	We start with Property (a). Consider any \texttt{visited} node $u$. If $u$ is not a node of $T_{v_i}$, then $T_u \subset T^*$. Thus, Property (a) holds for $u$ due to (1) and to the fact that $\Delta_{i-1}$ satisfies Property~(a). 
	If $u$ is a node of $T_{v_i}$, then the property holds due to (3).

	
	We show that $\Delta_i$ satisfies Properties (b.i)--(b.iv) of \cref{def:partially-canonical} by defining the sector $S_e$, for each edge $e=(u,v)$ where $v$ is the parent of $u$ and $v$ is \texttt{unvisited}, as the one defined for the edge $e$ in $\Delta_{i-1}$. Note that any \texttt{unvisited} node $v$ of $T[i]$ is also \texttt{unvisited} in $T[i-1]$.
	
	Property (b.i) holds since the \texttt{unvisited} nodes of $T[i]$, as well as their children, have the same placement in $\Delta_i$ and in $\Delta_{i-1}$ (by (1) and (2)) and since $\Delta_{i-1}$ satisfies Property (b.i).
	
	We show that Property (b.ii) holds. Consider any edge $e=(u,v)$ of $T[i]$, where $v$ is the parent of $u$ and $v$ is \texttt{unvisited}; note that $v$ is a node of $T^*$. By~(1) and since $\Delta_{i-1}$ satisfies Property (b.ii), we have that $S_e$ contains no node of $T^*$ with the exception of $v$ and of, possibly, the nodes of $T_u$, and contains no edge of $T^*$ with the exception of $(u,v)$ and of, possibly, the edges of $T_u$. If $u=v_i$ or if $u$ is a proper ancestor of $v_i$, then this completes the proof. Otherwise, it remains to argue that $S_e$ contains no node or edge of $T_{v_i}$; this follows from the fact that $S_e$ does not intersect $\largebox{v_i}$, by Property~(d) of $\Delta_{i-1}$ and by~(1) and~(2), and from the fact that the drawing of $T_{v_i}$ in $\Delta_i$ is canonical, hence by~\cref{remark:canonical-area} it is contained in $\smallbox{v_i}$ and thus in $\largebox{v_i}$.
	
	The drawing $\Delta_{i}$ trivially satisfies Properties (b.iii) and (b.iv), given that $\Delta_{i-1}$ satisfies the same properties, and given that the sector $S_e$ of any edge $e=(u,v)$ of $T[i]$, where $v$ is \texttt{unvisited} and $u$ is a child of $v$, is the same both in $\Delta_{i-1}$ and in $\Delta_{i}$, by~(1) and~(2) and by construction.   

	We show that $\Delta_i$ satisfies Property (c) of \cref{def:partially-canonical}.
	First, for any two \texttt{unvisited} nodes $v$ and $w$ of $T[i-1]$, it holds $\largebox{v} \cap \largebox{w} = \emptyset$ in $\Delta_{i-1}$, since $\Delta_{i-1}$ satisfies Property (c) of \cref{def:partially-canonical}.
	Further, any node that is \texttt{unvisited} in $T[i]$ is also \texttt{unvisited} in $T[i-1]$. 
	Finally, by~(1) the \texttt{unvisited} nodes of $T[i]$ have the same placement in $\Delta_i$ as in $\Delta_{i-1}$. Thus, Property (c) of \cref{def:partially-canonical} holds~for~$\Delta_i$.
	
	
	We show that $\Delta_i$ satisfies Property (d) of \cref{def:partially-canonical}. Consider any \texttt{unvisited} node $v$ of $T[i]$ and note that $v$ is in $T^*$. By~(1) the node $v$ has the same placement in $\Delta_i$ as in $\Delta_{i-1}$, hence the box $\largebox{v}$ is the same both in $\Delta_i$ and in $\Delta_{i-1}$. By~(1) and~(2) and since $\Delta_{i-1}$ satisfies Property (d) of \cref{def:partially-canonical}, it follows that $\largebox{v}$ contains no node of $T^*+(v_i,p(v_i))$ different from $v$, and any edge $e$ of $T^*+(v_i,p(v_i))$ or any sector $S_e$ of an edge $e$ of $T^*+(v_i,p(v_i))$ intersecting $\largebox{v}$ is such that $e$ is incident to $v$. Further, $\largebox{v}$ contains no node or edge of $T_{v_i}$ since the drawing of $T_{v_i}$ in $\Delta_i$ is contained in $\smallbox{v_i}$, by~\cref{remark:canonical-area}, since $\smallbox{v_i}$ is contained in $\largebox{v_i}$, by definition, and since $\largebox{v}\cap \largebox{v_i}=\emptyset$, by Property (c) of $\Delta_{i-1}$. Finally, note that no sector is defined for any edge of $T_{v_i}$, as all the nodes of $T_{v_i}$ are \texttt{visited} in $T[i]$.
\end{proof}

In order to prove \cref{lemma:partially-canonical-from-partially-canonical}, it remains to show how to construct an $O(\deg(v_i))$-step planar morph $\mathcal M_{i-1,i}$ from $\Delta_{i-1}$ to $\Delta_{i}$ such that $w(\mathcal M_{i-1,i}) \leq w(\Delta_0) + \ell_0 + 4n$ and $h(\mathcal M_{i-1,i}) \leq h(\Delta_0) + \ell_0 + 4n$. This is done in several phases as follows.
Since the drawing of $T^*$ stays unchanged during $\mathcal M_{i-1,i}$, no two edges in $T^*$ cross during such a morph. Thus, we will omit to repeat this in the proofs of the planarity of the morphing steps that compose $\mathcal M_{i-1,i}$.

\begin{figure}[tb] 
	\centering
	\subfloat[]
	{\includegraphics[width=.5\textwidth,page=1]{boxed.pdf}}
	\hfil
	\subfloat[]
	{\includegraphics[width=.5\textwidth,page=2]{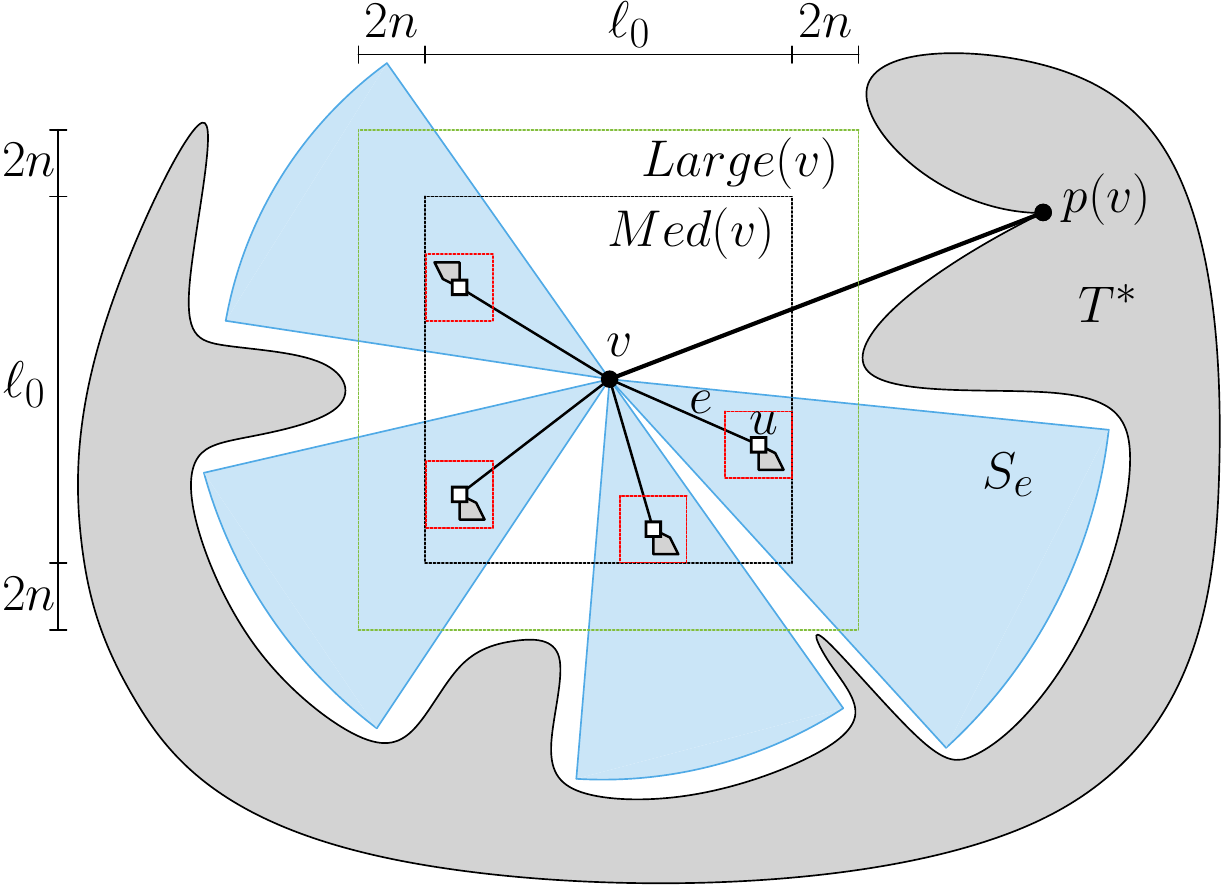}}
	\caption{
		(a) A partially-canonical drawing $\Delta_{i-1}$ of the tree $T[i-1]$; the subtree $T^*$ lies in the gray region, \texttt{visited} and \texttt{unvisited} nodes are represented as squares and circles, respectively. 
		(b) The drawing $\Delta'$ of the morph $\langle \Delta_{i-1},\Delta'\rangle$ of \cref{claim:m-prime}.}
\label{fig:partially-canonical-first-morph}
\end{figure}

First, consider the drawing $\Delta'$ of $T$ obtained as described next; refer to \cref{fig:partially-canonical-first-morph}.
Initialize $\Delta'=\Delta_{i-1}$. 
Then, for each child $u$ of $v_i$, translate the drawing of $T_{u}$ (which is an upward or downward canonical drawing, by Property~(a) of $\Delta_{i-1}$) so that $u$ is at the center of a $2n$-box $B_u$ that lies in the intersection between $S_e$ and $\mediumbox{v_i}$, whose corners have integer coordinates, and whose center $c_u$ is such that $y_{\Delta_{i-1}}(c_u) \leq y_{\Delta_{i-1}}(v_i)$ if and only if $y_{\Delta_{i-1}}(u) \leq y_{\Delta_{i-1}}(v_i)$; such a box exists by Property~(b.iii) of $\Delta_{i-1}$. 
%
%

\begin{claim}\label{claim:m-prime}
	The	morph $\langle \Delta_{i-1},\Delta'\rangle$ is planar.
\end{claim}

\begin{proof}
	First, for each child $u$ of $v_i$, the drawing of $T_u$ lies in the interior of the sector $S_e$ with $e = (v_i,u)$ both in $\Delta_{i-1}$ and in $\Delta'$; the former follows by Property~(b.i) of $\Delta_{i-1}$ and by \cref{remark:canonical-area}, while the latter follows by Property~(b.iii) of $\Delta_{i-1}$, by \cref{remark:canonical-area}, and by construction. It follows that the drawing of $T_u+(v_i,u)$ lies in the interior of $S_e$ throughout the morph $\langle \Delta_{i-1},\Delta'\rangle$. Thus, by Properties~(b.ii) and~(b.iv) of $\Delta_{i-1}$, the drawing of $T_u+(v_i,u)$ crosses neither the drawing of $T^*+(v_i,p(v_i))$ nor the drawing of $T_w+(v_i,w)$ during $\langle \Delta_{i-1},\Delta'\rangle$, where $w\neq u$ is a child of $v_i$.
	
	It remains to prove that no two edges in $T_u + (v_i,u)$ cross each other during $\langle \Delta_{i-1},\Delta'\rangle$, for each child $u$ of $v_i$. By construction, the drawing of $T_u$ in $\Delta'$ is a translation of the drawing of $T_u$ in $\Delta_{i-1}$, hence no two edges of $T_u$ cross during the morph.
	Further, by construction, we have $y_{\Delta_{i-1}}(c_u) \leq y_{\Delta_{i-1}}(v_i)$ if and only if $y_{\Delta_{i-1}}(u) \leq y_{\Delta_{i-1}}(v_i)$. This implies that $v_i$ lies above $u$ in $\Delta'$ if and only if it lies above $u$ in $\Delta_{i-1}$. Therefore, the edge $(v_i,u)$ does not cross any edge of $T_u$ during the morph. This concludes the proof that $\langle \Delta_{i-1},\Delta'\rangle$ is planar.
\end{proof}

\begin{figure}[t]
	\centering
	\includegraphics[page=7]{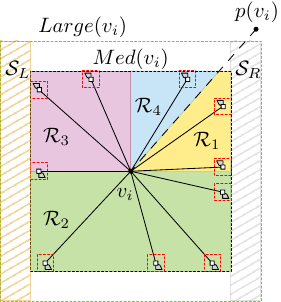}\hfil
	\caption{Regions for $v_i$.}
	\label{fig:regions}
\end{figure}

Second, we show how to move the subtrees rooted at the children of $v_i$ in the interior of $\largebox{v_i}$, so that they land in the position they have in $\Delta_i$. By Property~(d) of $\Delta_{i-1}$, no node or edge of $T^*$ intersects $\largebox{v_i}$. Since the drawing of $T^*+(v_i,p(v_i))$ stays unchanged throughout the morph from $\Delta'$ to $\Delta_i$, no node or edge of $T^*$ crosses any node or edge of $T_{v_i}+(v_i,p(v_i))$ during such a morph. Thus, in the proofs of the planarity of the morphing steps that compose the morph from $\Delta'$ to $\Delta_i$, we will only need to prove that the subtrees rooted at the children of $v_i$ do not cross each other and do not cross the edges incident to $v_i$.

The way we move the subtrees rooted at the children of $v_i$ depends on their placement with respect to $v_i$ and to the drawing of the edge $(v_i,p(v_i))$. We assume that, if $v_i\neq r(T)$, then $y(p(v_i))\geq y(v_i)$ and $x(p(v_i)) \geq x(v_i)$; the other cases can be treated similarly. We distinguish four regions $\mathcal R_1$, $\mathcal R_2$, $\mathcal R_3$, and $\mathcal R_4$ defined as follows; refer to \cref{fig:regions}. Let $h_\rightarrow(v)$ and $h_\leftarrow(v)$ be the horizontal rays originating at a node $v$ and directed rightward and leftward, respectively. Further, let $h_\uparrow(v)$ and $h_\downarrow(v)$ be the vertical rays originating at a node $v$ and directed upward and downward, respectively.

\begin{description}
	\item[Region $\mathbf{\mathcal R_1}$] is defined as follows. If $v_i=r(T)$, then $\mathcal R_1$ is the intersection of $\mediumbox{v_i}$ with the wedge centered at $v_i$ obtained by clockwise rotating $h_\uparrow(v_i)$ until it coincides with $h_\rightarrow(v_i)$. Otherwise, $\mathcal R_1$ is the intersection of $\mediumbox{v_i}$ with the wedge centered at $v_i$ obtained by counter-clockwise rotating $h_\rightarrow(v_i)$ until it passes~through~$p(v_i)$; note that, if $(v_i,p(v_i))$ is a horizontal segment, then $\mathcal R_1=\emptyset$. 
	\item[Region $\mathbf{\mathcal R_2}$] is the rectangular region that is the lower half of $\mediumbox{v_i}$;
	\item[Region $\mathbf{\mathcal R_3}$] is the intersection of $\mediumbox{v_i}$ with the wedge centered at $v_i$ obtained by clockwise rotating $h_\leftarrow(v_i)$ until it coincides with $h_\uparrow(v_i)$;~and
	\item[Region $\mathbf{\mathcal R_4}$] is defined as follows. If $v_i=r(T)$, then $\mathcal R_4=\emptyset$. Otherwise, $\mathcal R_4$ is the intersection of $\mediumbox{v_i}$ with the wedge centered at $v_i$ obtained by clockwise rotating $h_\uparrow(v_i)$ until it passes~through~$p(v_i)$; note that, if $(v_i,p(v_i))$ is a vertical segment, then $\mathcal R_4=\emptyset$.
\end{description}

Note that $\mediumbox{v_i}=\mathcal R_1 \cup \mathcal R_2 \cup \mathcal R_3 \cup \mathcal R_4$.

\begin{figure}[tb!]
	\centering
	\subfloat[$\Delta'$\label{m:lower-a}]{\includegraphics[page=3,height=.22\textwidth]{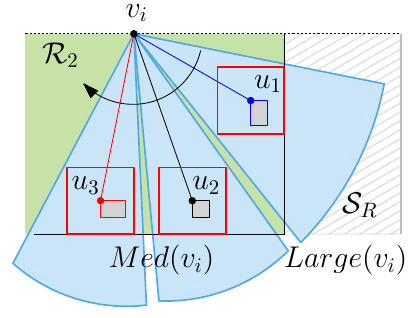}}\hfil
	\subfloat[$\Psi_1$\label{m:lower-b}]{\includegraphics[page=4,height=.22\textwidth]{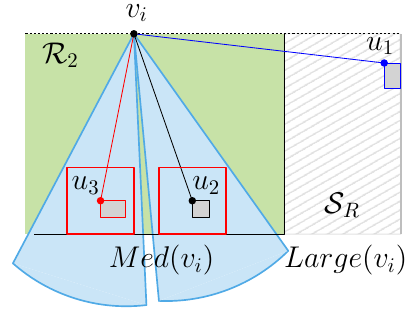}}\\
	\subfloat[$\Psi_3$\label{m:lower-c}]{\includegraphics[page=5,height=.22\textwidth]{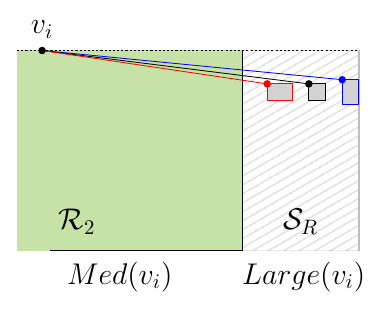}}\hfil
	\subfloat[$\Delta^+$\label{m:lower-d}]{\includegraphics[page=6,height=.22\textwidth]{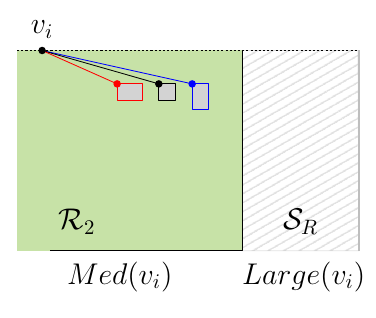}}
	\caption{Illustrations for the proof of \cref{lemma:partially-canonical-from-partially-canonical}.}
	\label{m:lower}
\end{figure}

We define two more regions, which will be exploited as ``buffer regions'' to allow rotations of subtrees via \cref{lemma:rotate} in a safe way; refer again to \cref{fig:regions}. Let $\mathcal S_L$ and $\mathcal S_R$ be the rectangular regions in $\Delta'$ containing all the points in $\largebox{v_i} - \mediumbox{v_i}$ to the left of the left side of $\mediumbox{v_i}$ and to the right of the right side of $\mediumbox{v_i}$, respectively. By Properties~(b.iii) and~(d) of the partially-canonical drawing $\Delta_{i-1}$, by the construction of $\Delta'$, and by the assumptions that, if $v_i\neq r(T)$, then $y(p(v_i))\geq y(v_i)$ and $x(p(v_i)) \geq x(v_i)$, we have that $\mathcal S_L$ is empty, while $\mathcal S_R$ may only contain the drawing of the part of the edge $(v_i, p(v_i))$ that possibly traverses such a region.

We start by dealing with the children $u_j$ of $v_i$ that lie in the interior of $\mathcal R_2$; refer to \cref{m:lower}.
Consider the edges $(v_i,u_j)$ in the order $(v_i,u_1),(v_i,u_2),\dots, (v_i,u_m)$ in which such edges are encountered while clockwise rotating $h_\rightarrow(v_i)$; see \cref{m:lower-a}.
Let $\Psi_1$ be the drawing obtained from $\Delta'$ by translating the drawing of the tree $T_{u_1}$ so that $u_1$ lies one unit below $v_i$ and so that the right side of the bounding box of the drawing of $T_{u_1}$ lies upon the right side of $\largebox{v_i}$; refer to \cref{m:lower-b}. 
%
%
\begin{claim}
	The morph $\langle \Delta',\Psi_1\rangle$ is planar.
\end{claim}
\begin{proof}
	By \cref{remark:canonical-area}, we have that the drawing of $T_{u_1}$ in $\Psi_1$ lies in $\mathcal S_R$, hence the drawing of $T_{u_1}$ is in $\largebox{v_i}$ throughout $\langle \Delta',\Psi_1\rangle$. Note that only $T_{u_1}+(v_i,u_1)$ moves during $\langle \Delta',\Psi_1\rangle$, hence any crossing during such a morph involves an edge of $T_{u_1}+(v_i,u_1)$ and an edge of $T_{v_i}+(v_i,p(v_i))$. 
	
	First, no two edges of $T_{u_1}$ cross each other during $\langle \Delta',\Psi_1\rangle$ since the drawing of $T_{u_1}$ in $\Psi_1$ is a translation of the drawing of $T_{u_1}$ in $\Delta'$.
	
	Second, since the edge $(v_i,u_1)$ lies above the drawing of $T_{u_1}$ both in $\Delta'$ and in $\Psi_1$, it lies above the drawing of $T_{u_1}$ throughout $\langle \Delta',\Psi_1\rangle$, hence it does not cross any edge of $T_{u_1}$. 
	
	It remains to prove that no edge of $T_{u_1}+(v_i,u_1)$ crosses an edge of $T_{v_i}+(v_i,p(v_i))$ that is not in $T_{u_1}+(v_i,u_1)$. Consider the region $R(u_1)$ which is the intersection of $\largebox{v_i}$ and the wedge obtained by clockwise rotating the ray $h_\rightarrow(v_i)$ around $v_i$ until it passes through both rays delimiting the sector $S_{v_i,u_1}$. We have that $T_{u_1}$ moves in the interior of $R(u_1)$ during the morph. Further, by Properties~(b.iii) and~(b.iv) of the partially-canonical drawing $\Delta_{i-1}$, and by the assumptions that $y(p(v_i))\geq y(v_i)$, that $x(p(v_i))\geq x(v_i)$, and that $(v_i,u_1)$ is the first edge incident to $v_i$ that is encountered while clockwise rotating $h_\rightarrow(v_i)$, it follows that $R(u_1)$ does not contain any edge of $T_{v_i}+(v_i,p(v_i))$ that is not in $T_{u_1}+(v_i,u_1)$ throughout the morph $\langle \Delta',\Psi_1\rangle$. Hence, no edge of $T_{u_1}+(v_i,u_1)$ crosses an edge of $T_{v_i}+(v_i,p(v_i))$ that is not in $T_{u_1}+(v_i,u_1)$. 
\end{proof}

For $j=2,\dots,m$, let $\Psi_j$ be the drawing obtained from $\Psi_{j-1}$ by translating the drawing of the tree $T_{u_j}$ so that $u_j$ lies one unit below $v_i$ and so that the right side of the bounding box of the drawing of $T_{u_j}$ lies one unit to the left of $u_{j-1}$; refer to \cref{m:lower-c}.

\begin{claim}
	For $j=2,\dots,m$, the morph $\langle \Psi_{j-1},\Psi_j\rangle$ is planar.
\end{claim}

\begin{proof}
By \cref{remark:canonical-area}, we have that the drawing of $T_{u_j}$ in $\Psi_j$ lies in $\mathcal S_R$, hence the drawing of $T_{u_j}$ is in $\largebox{v_i}$ throughout $\langle \Psi_{j-1},\Psi_j\rangle$. Note that only $T_{u_j}+(v_i,u_j)$ moves during $\langle \Psi_{j-1},\Psi_j\rangle$, hence any crossing during such a morph involves an edge of $T_{u_j}+(v_i,u_j)$ and an edge of $T_{v_i}+(v_i,p(v_i))$. 	

First, no two edges of $T_{u_j}$ cross each other during $\langle \Psi_{j-1},\Psi_j\rangle$ since the drawing of $T_{u_j}$ in $\Psi_j$ is a translation of the drawing of $T_{u_j}$ in $\Psi_{j-1}$.

Second, since the edge $(v_i,u_j)$ lies above the drawing of $T_{u_j}$ both in $\Psi_{j-1}$ and in $\Psi_j$, it lies above the drawing of $T_{u_j}$ throughout $\langle \Psi_{j-1},\Psi_j\rangle$, hence it does not cross any edge of $T_{u_j}$. 

It remains to prove that no edge of $T_{u_j}+(v_i,u_j)$ crosses an edge of $T_{v_i}+(v_i,p(v_i))$ that is not in $T_{u_j}+(v_i,u_j)$. First, by the assumption that $y(p(v_i))\geq y(v_i)$ and since the drawing of $T_{u_j}+(v_i,u_j)$  stays in the lower half of $\largebox{v_i}$ throughout $\langle \Psi_{j-1},\Psi_j\rangle$, we have that the drawing of $T_{u_j}+(v_i,u_j)$ does not cross the edge $(v_i,p(v_i))$, if such an edge exists. Second, we deal with possible crossings between $T_{u_j}+(v_i,u_j)$ and $T_{u_l}+(v_i,u_l)$, for any $l=j+1,\dots,m$. Consider the region $R(u_j)$ which is the intersection of $\largebox{v_i}$ and the wedge obtained by clockwise rotating the ray $h_\rightarrow(v_i)$ around $v_i$ until it passes through both rays delimiting the sector $S_{v_i,u_j}$. We have that $T_{u_j}$ moves in the interior of $R(u_j)$ during the morph. By Properties~(b.iii) and~(b.iv) of the partially-canonical drawing $\Delta_{i-1}$, the region $R(u_j)$ does not contain any edge of $T_{u_l}+(v_i,u_l)$ throughout $\langle \Psi_{j-1},\Psi_j\rangle$, hence no edge of $T_{u_j}+(v_i,u_j)$ crosses any edge of $T_{u_l}+(v_i,u_l)$. Third, we deal with possible crossings between $T_{u_j}+(v_i,u_j)$ and $T_{u_l}+(v_i,u_l)$, for any $l=1,\dots,j-1$. The edge $(v_i,u_l)$ is above the drawing of $T_{u_j}+(v_i,u_j)$ throughout $\langle \Psi_{j-1},\Psi_j\rangle$, hence $(v_i,u_l)$ does not cross any edge of $T_{u_j}+(v_i,u_j)$. Finally, the drawing of $T_{u_l}$ is to the right of the drawing of $T_{u_j}+(v_i,u_j)$ throughout $\langle \Psi_{j-1},\Psi_j\rangle$, hence no edge of $T_{u_l}$ crosses any edge of $T_{u_j}+(v_i,u_j)$.
\end{proof}

Let $\Delta^+$ be the drawing obtained from $\Psi_m$ by horizontally translating every subtree $T_{u_j}$ so that $u_j$ lands at the position it has in $\Delta_i$, for $j=1,2,\dots,m$; see \cref{m:lower-c,m:lower-d}.
%
%
\begin{claim}\label{claim:horizontal-shift}
	The morph $\langle \Psi_m,\Delta^+\rangle$ is planar.
\end{claim}

\begin{proof}
Let $T_2$ be the rooted ordered tree composed of the node $v_i$ and of $T_{u_j}+(v_i,u_j)$, for $j=1,\dots,m$; that is, $T_2$ is the tree obtained from $T_{v_i}$ by removing the nodes of the subtrees rooted at the children of $v_i$ in $\mathcal R_1$, $\mathcal R_3$, and $\mathcal R_4$, and their incident edges. By \cref{remark:canonical-area}, we have that the drawing of $T_2$ in $\Delta^+$ lies in $\smallbox{v_i}$, hence the drawing of $T_2$ is in $\largebox{v_i}$ throughout $\langle \Psi_m,\Delta^+\rangle$. Note that only the nodes of $T_2$ move during $\langle \Psi_m,\Delta^+\rangle$, hence any crossing during such a morph involves an edge of $T_2$ and an edge of $T_{v_i}+(v_i,p(v_i))$. 	

By the assumption that $y(p(v_i))\geq y(v_i)$ and since the drawing of $T_2$  stays in the lower half of $\largebox{v_i}$ throughout $\langle \Psi_m,\Delta^+\rangle$, we have that the drawing of $T_2$ does not cross the edge $(v_i,p(v_i))$, if such an edge exists, and does not cross any edge of a subtree of $T_{v_i}$ rooted at a child of $v_i$ that lies in $\mathcal R_1$, $\mathcal R_3$, or $\mathcal R_4$. Finally, the drawing of $T_2$ is order-preserving, strictly-upward, straight-line, and planar both in $\Psi_m$ and in $\Delta^+$; further, by construction, we have $y_{\Psi_m}(v)=y_{\Delta^+}(v)$, for each node $v$ of $T_2$. Thus, the linear morph $\langle \Psi_m,\Delta^+\rangle$ is planar, by \cref{le:unidirectional-morph}.
\end{proof}

\begin{figure}[tb!]
	\centering
	\subfloat[$\xi_0^{4} = \Delta^+$\label{fig:r-two-a}]{\includegraphics[page=8,width=.22\textwidth]{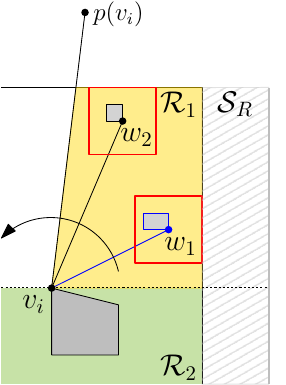}}\hfil
	\subfloat[$\xi_1^{1}$\label{fig:r-two-b}]{\includegraphics[page=9,width=.22\textwidth]{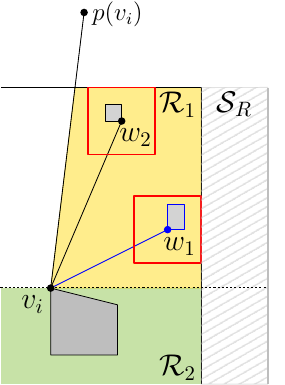}}\hfil
	\subfloat[$\xi_1^{2}$\label{fig:r-two-c}]{\includegraphics[page=10,width=.22\textwidth]{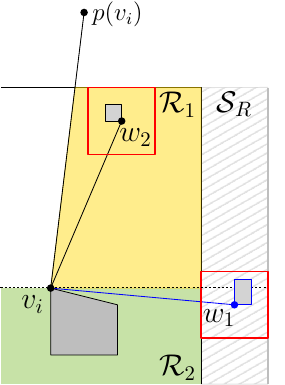}}\\
	\subfloat[$\xi_1^{3}$\label{fig:r-two-d}]{\includegraphics[page=11,width=.22\textwidth]{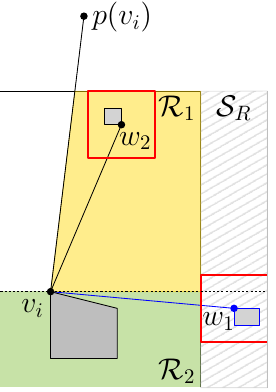}}\hfil
	\subfloat[$\xi_1^{4}$\label{fig:r-two-e}]{\includegraphics[page=12,width=.22\textwidth]{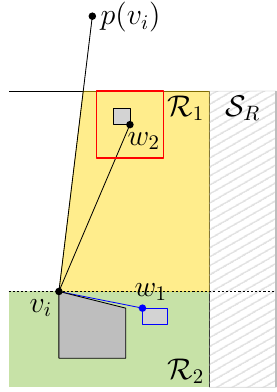}}
	
	\caption{Illustrations for \cref{lemma:partially-canonical-from-partially-canonical}, focused on the children of $v_i$ that lie in $\mathcal R_1$.}
	\label{fig:r-two}
\end{figure}

Next, we deal with the children $w_j$ of $v_i$ that lie in the interior of $\mathcal R_1$.
Consider the edges $(v_i,w_j)$ in the order $(v_i,w_1),(v_i,w_2),\dots, (v_i,w_\ell)$ in which such edges are encountered while counter-clockwise rotating $h_{\rightarrow}(v_i)$ around $v_i$; refer to \cref{fig:r-two}.
We are going to move the subtrees rooted at the children of $v_i$ in~$\mathcal R_1$, one by one in the order $T_{w_1},T_{w_2},\dots,T_{w_\ell}$, so that they land in the position that they have in~$\Delta_i$. Such a movement consists of four phases. 
First, we rotate the drawing of $T_{w_j}$ so that it becomes leftward canonical (see \cref{fig:r-two-b}). Second, we translate the drawing of $T_{w_j}$ so that $w_j$ lies
in the interior of $S_R$ and one unit below $v_i$ (see \cref{fig:r-two-c}).
Third, we rotate the drawing of $T_{w_j}$ so that it becomes upward canonical (see \cref{fig:r-two-d}).
Finally, we horizontally translate the drawing of $T_{w_j}$ to its final position in $\Delta_i$ (see \cref{fig:r-two-e}). We now provide the details of the above transformations.

For $j=1,\dots,\ell$, let $\xi_{j-1}^{4}$ be a drawing of $T$ with the following properties, where $\xi_0^{4} = \Delta^+$ (refer to \cref{fig:r-two-a,fig:r-two-e}, showing $\xi_{0}^{4}$ and $\xi_{1}^{4}$, respectively):
\begin{enumerate}[(P1)]
	\item \label{item:props-four-ii} the drawing of $T^*$ is the same as in $\Delta_i$; 
	\item \label{item:props-four-iii} $v_i$ lies at the same point as in $\Delta_i$;  
	\item \label{item:props-four-iv} the drawing of the subtrees  $T_{u_1},T_{u_2}, \dots,T_{u_m}$ and $T_{w_1},T_{w_2}, \dots,T_{w_{j-1}}$ is the same as in $\Delta_i$;
	\item \label{item:props-four-i} the drawing of the subtrees $T_{w_{j}},T_{w_{j+1}}, \dots, T_{w_\ell}$ is the same as~in~$\Delta^+$; and
	\item \label{item:props-four-v} the drawing of the subtrees of $T_{v_i}$ rooted at the children of $v_i$ that lie in the interior of $\mathcal R_3$ and $\mathcal R_4$ is the same as in $\Delta^+$. 
\end{enumerate}

For $j=1,\dots,\ell$, we construct a drawing $\xi_j^{1}$ from $\xi_{j-1}^{4}$ by rotating $T_{w_j}$ so that it is leftward canonical in $\xi_j^{1}$ and by leaving the position of the nodes not in $T_{w_j}$ unaltered. This rotation can be accomplished via a linear morph $\langle \xi_{j-1}^{4}, \xi_j^{1}\rangle$ by \cref{lemma:rotate}. 

\begin{claim}
	For $j=1,\dots,\ell$, the morph $\langle \xi_{j-1}^{4}, \xi_j^{1}\rangle$ is planar.
\end{claim}

\begin{proof}
By \cref{remark:canonical-area}, we have that the drawing of $T_{w_j}$ in $\xi_j^{1}$ lies in $\smallbox{w_j}$, hence the drawing of $T_{w_j}$ is in $\largebox{v_i}$ throughout $\langle \xi_{j-1}^{4}, \xi_j^{1}\rangle$. Note that only $T_{w_j}+(v_i,w_j)$ moves during $\langle \xi_{j-1}^{4}, \xi_j^{1}\rangle$, hence any crossing during such a morph involves an edge of $T_{w_j}+(v_i,w_j)$ and an edge of $T_{v_i}+(v_i,p(v_i))$. 	

First, no two edges of $T_{w_j}$ cross each other during $\langle \xi_{j-1}^{4}, \xi_j^{1}\rangle$, by \cref{lemma:rotate}.

Second, since the edge $(v_i,w_j)$ lies below the drawing of $T_{w_j}$ both in $\xi_{j-1}^{4}$ and in $\xi_j^{1}$, it lies below the drawing of $T_{w_j}$ throughout $\langle \xi_{j-1}^{4}, \xi_j^{1}\rangle$, hence it does not cross any edge of $T_{w_j}$. 

It remains to prove that no edge of $T_{w_j}+(v_i,w_j)$ crosses an edge of $T_{v_i}+(v_i,p(v_i))$ that is not in $T_{w_j}+(v_i,w_j)$. No edge of $T_{w_j}+(v_i,w_j)$ crosses any edge of $T_{u_h}+(v_i,u_h)$, for any $h\in \{1,\dots,m\}$, and any edge of $T_{w_h}+(v_i,w_h)$, for any $h\in \{1,\dots,j-1\}$, by Property~(P3). Further, no edge of $T_{w_j}+(v_i,w_j)$ crosses any edge of $T_{w_h}+(v_i,w_h)$, for any $h\in \{j+1,\dots,\ell\}$, by Properties~(P4) and~(P5) and since the sector $S_{(v_i,w_j)}$, which contains $\smallbox{w_j}$ (by Property~(b.i) of $\Delta_{i-1}$) and hence the drawing of $T_{w_j}+(v_i,w_j)$ throughout the morph $\langle \xi_{j-1}^{4}, \xi_j^{1}\rangle$, does not intersect the sector $S_{(v_i,w_h)}$ (by Property~(b.iv) of $\Delta_{i-1}$), which contains $\smallbox{w_h}$ (by Property~(b.i) of $\Delta_{i-1}$) and hence the drawing of $T_{w_h}+(v_i,w_h)$ throughout the morph $\langle \xi_{j-1}^{4}, \xi_j^{1}\rangle$. The same arguments prove that no edge of $T_{w_j}+(v_i,w_j)$ crosses any edge of $T_{w_h}+(v_i,w_h)$, for any $h\in \{j+1,\dots,\ell\}$, crosses any edge of a subtree of $T_{v_i}$ rooted at a child of $v_i$ lying in $\mathcal R_3$ or in $\mathcal R_4$. 
\end{proof}

For $j=1,\dots,\ell$, let $\xi_j^{2}$ be the drawing obtained from $\xi_j^{1}$ by translating the drawing of $T_{w_j}$ in such a way that $w_j$ lies one unit below $v_i$ and so that the right side of $\smallbox{w_j}$ lies upon the right side of $\largebox{v_i}$. 

\begin{claim}
	For $j=1,\dots,\ell$, the morph $\langle \xi^1_j,\xi^2_j\rangle$ is planar.
\end{claim}

\begin{proof}
By \cref{remark:canonical-area}, we have that the drawing of $T_{w_j}$ in $\xi_j^{1}$ lies in $\smallbox{w_j}$, hence the drawing of $T_{w_j}$ is in $\largebox{v_i}$ throughout $\langle \xi_{j}^{1}, \xi_j^{2}\rangle$. Note that only $T_{w_j}+(v_i,w_j)$ moves during $\langle \xi_{j}^{1}, \xi_j^{2}\rangle$, hence any crossing during such a morph involves an edge of $T_{w_j}+(v_i,w_j)$ and an edge of $T_{v_i}+(v_i,p(v_i))$. 	

First, no two edges of $T_{w_j}$ cross each other during $\langle \xi_{j}^{1}, \xi_j^{2}\rangle$ since the drawing of $T_{w_j}$ in $\xi_j^{2}$ is a translation of the drawing of $T_{w_j}$ in $\xi_{j}^{1}$.

Second, since the edge $(v_i,w_j)$ lies to the left of the drawing of $T_{w_j}$ both in $\xi_{j}^{1}$ and in $\xi_j^{2}$, it lies to the left of the drawing of $T_{w_j}$ throughout $\langle \xi_{j}^{1}, \xi_j^{2}\rangle$, hence it does not cross any edge of $T_{w_j}$. 

It remains to prove that no edge of $T_{w_j}+(v_i,w_j)$ crosses an edge of $T_{v_i}+(v_i,p(v_i))$ that is not in $T_{w_j}+(v_i,w_j)$. First, consider the region $R(w_j)$ which is the intersection of $\largebox{w_j}$ and the wedge obtained by counter-clockwise rotating the ray $h_\downarrow(w_j)$ around $w_j$ until it passes through both rays delimiting the sector $S_{v_i,w_j}$. We have that $T_{w_j}$ moves in the interior of $R(w_j)$ during the morph. Further, by Property~(P4), the region $R(w_j)$ does not contain any edge of $T_{w_h}+(v_i,w_h)$, for any $h\in \{j+1,\dots,\ell\}$, hence no edge of $T_{w_j}+(v_i,w_j)$ crosses any edge of $T_{w_h}+(v_i,w_h)$, for any $h\in \{j+1,\dots,\ell\}$. Similarly, no edge of $T_{w_j}+(v_i,w_j)$ crosses $(v_i,p(v_i))$ and, by Property~(P5), no edge of $T_{w_j}+(v_i,w_j)$ crosses any edge of a subtree of $T_{v_i}$ rooted at a child of $v_i$ lying in $\mathcal R_3$ or in $\mathcal R_4$. Next, we argue about the absence of crossings between the edges of $T_{w_j}+(v_i,w_j)$ and the edges of $T_{w_h}+(v_i,w_h)$, for any $h\in \{1,\dots,j-1\}$. The edge $(v_i,w_j)$ lies above all the edges of $T_{w_h}+(v_i,w_h)$ throughout $\langle \xi_{j}^{1}, \xi_j^{2}\rangle$. Further, $T_{w_j}$ lies above $T_{w_h}$ in every drawing of the morph $\langle \xi_{j}^{1}, \xi_j^{2}\rangle$, except for $\xi_j^{2}$, in which $T_{w_j}$ is to the right of $T_{w_h}$. Finally, $T_{w_j}$ lies above the line through $(v_i,w_h)$ in $\xi_{j}^{1}$ and in $\xi_j^{2}$, and hence throughout $\langle \xi_{j}^{1}, \xi_j^{2}\rangle$. It follows that no edge of $T_{w_j}+(v_i,w_j)$ crosses any edge of $T_{w_h}+(v_i,w_h)$, for any $h\in \{1,\dots,j-1\}$. An analogous proof shows that  no edge of $T_{w_j}+(v_i,w_j)$ crosses any edge of $T_{u_h}+(v_i,u_h)$, for any $h\in \{1,\dots,m\}$. 
\end{proof}

For $j=1,\dots,\ell$, we construct a drawing $\xi_j^{3}$ from $\xi_j^{2}$ by rotating $T_{w_j}$ so that it is upward canonical in $\xi_j^{3}$ and by leaving the position of the nodes not in $T_{w_j}$ unaltered. This rotation can be accomplished via a linear morph $\langle \xi_j^{2}, \xi_j^{3}\rangle$ by \cref{lemma:rotate}. 

\begin{claim}
	For $j=1,\dots,\ell$, the morph $\langle \xi_{j}^{2}, \xi_j^{3}\rangle$ is planar.
\end{claim}
\begin{proof}
By \cref{remark:canonical-area}, we have that the drawing of $T_{w_j}$ lies in $\smallbox{w_j}$ both in $\xi_j^{2}$ and in $\xi_j^{3}$, and hence throughout $\langle \xi_{j}^{2}, \xi_j^{3}\rangle$. This implies that the drawing of $T_{w_j}$ lies in $\largebox{v_i}$, and in particular in $\mathcal S_R$, throughout $\langle \xi_{j}^{2}, \xi_j^{3}\rangle$. Note that only $T_{w_j}$ moves during $\langle \xi_{j}^{2}, \xi_j^{3}\rangle$, hence any crossing during such a morph involves an edge of $T_{w_j}$ and an edge of $T_{v_i}+(v_i,p(v_i))$. 

First, no two edges of $T_{w_j}$ cross each other during $\langle \xi_{j}^{2}, \xi_j^{3}\rangle$, by \cref{lemma:rotate}.

Second, since the edge $(v_i,w_j)$ lies to the left of the drawing of $T_{w_j}$ both in $\xi_{j}^{2}$ and in $\xi_j^{3}$, it lies to the left of the drawing of $T_{w_j}$ throughout $\langle \xi_{j}^{2}, \xi_j^{3}\rangle$, hence it does not cross any edge of $T_{w_j}$. 

Finally, no edge of $T_{w_j}$ crosses any edge of $T_{v_i}+(v_i,p(v_i))$ that is not in $T_{w_j}+(v_i,w_j)$, as throughout $\langle \xi_{j}^{2}, \xi_j^{3}\rangle$ the former lies in $\mathcal S_R$, while the latter lies in $\mediumbox{v_i}$.
\end{proof}

For $j=1,\dots,\ell$, let $\xi_j^{4}$ be the drawing obtained from $\xi_j^{3}$ by translating the drawing of $T_{w_j}$ in such a way that $w_j$ lands at the position it has in $\Delta_i$ (that is, one unit below $v_i$ and one unit to the right of the rightmost node in $T_{w_{j-1}}$, if $j\geq 2$, or one unit to the right of the rightmost node in $T_{u_m}$, if $j=1$ and $v_i$ has children in $\mathcal R_2$ in $\Delta'$, or one unit to the right of $v_i$ if $j=1$ and $v_i$ has no child in $\mathcal R_2$ in $\Delta'$). 
%
%

\begin{claim}
	For $j=1,\dots,\ell$, the morph $\langle \xi_{j}^{3}, \xi_j^{4}\rangle$ is planar.
\end{claim}

\begin{proof}
By \cref{remark:canonical-area}, we have that the drawing of $T_{w_j}$ in $\xi_j^{3}$ lies in $\smallbox{w_j}$, hence the drawing of $T_{w_j}$ is in $\largebox{v_i}$ throughout $\langle \xi_{j}^{3}, \xi_j^{4}\rangle$. Note that only $T_{w_j}+(v_i,w_j)$ moves during $\langle \xi_{j}^{3}, \xi_j^{4}\rangle$, hence any crossing during such a morph involves an edge of $T_{w_j}+(v_i,w_j)$ and an edge of $T_{v_i}+(v_i,p(v_i))$. 	

First, no two edges of $T_{w_j}$ cross each other during $\langle \xi_{j}^{3}, \xi_j^{4}\rangle$ since the drawing of $T_{w_j}$ in $\xi_j^{4}$ is a translation of the drawing of $T_{w_j}$ in $\xi_{j}^{3}$.

Second, since the edge $(v_i,w_j)$ lies above the drawing of $T_{w_j}$ both in $\xi_{j}^{3}$ and in $\xi_j^{4}$, it lies above the drawing of $T_{w_j}$ throughout $\langle \xi_{j}^{3}, \xi_j^{4}\rangle$, hence it does not cross any edge of $T_{w_j}$. 

It remains to prove that no edge of $T_{w_j}+(v_i,w_j)$ crosses an edge of $T_{v_i}+(v_i,p(v_i))$ that is not in $T_{w_j}+(v_i,w_j)$. First, the edges of $T_{w_j}+(v_i,w_j)$ lie below the horizontal line through $v_i$ throughout $\langle \xi_{j}^{3}, \xi_j^{4}\rangle$. Hence, no edge of $T_{w_j}+(v_i,w_j)$ crosses the edge $(v_i,p(v_i))$, or the edges of $T_{w_h}+(v_i,w_h)$, for any $h\in \{j+1,\dots,\ell\}$, or the edges of the subtrees of $T_{v_i}$ rooted at the children of $v_i$ lying in $\mathcal R_3$ or in $\mathcal R_4$, as all such edges lie above the horizontal line through $v_i$ throughout $\langle \xi_{j}^{3}, \xi_j^{4}\rangle$. Next, we argue about the absence of crossings between the edges of $T_{w_j}+(v_i,w_j)$ and the edges of $T_{w_h}+(v_i,w_h)$, for any $h\in \{1,\dots,j-1\}$. The edge $(v_i,w_j)$ lies above all the edges of $T_{w_h}+(v_i,w_h)$ throughout $\langle \xi_{j}^{3}, \xi_j^{4}\rangle$. Further, $T_{w_j}$ lies to the right of $T_{w_h}+(v_i,w_h)$ throughout $\langle \xi_{j}^{3}, \xi_j^{4}\rangle$. It follows that no edge of $T_{w_j}+(v_i,w_j)$ crosses any edge of $T_{w_h}+(v_i,w_h)$, for any $h\in \{1,\dots,j-1\}$. An analogous proof shows that  no edge of $T_{w_j}+(v_i,w_j)$ crosses any edge of $T_{u_h}+(v_i,u_h)$, for any $h\in \{1,\dots,m\}$. 
\end{proof}

Note that $\xi_j^{4}$ satisfies Properties (P1)--(P5), given that $\xi_{j-1}^{4}$ satisfies the same properties and given that during the morph $\langle \xi_{j-1}^{4}, \xi_{j}^{1}, \xi_{j}^{2}, \xi_{j}^{3}, \xi_j^{4}\rangle$ only the nodes of $T_{w_j}$ move, from their position in $\Delta^+$ to their position in $\Delta_i$. Eventually, the drawing $\xi_{\ell}^{4}$ coincides with $\Delta_{i}$, except for the drawing of the subtrees lying in the interior of $\mathcal R_3$ and $\mathcal R_4$.

Subtrees in $\mathcal R_3$ are treated symmetrically to the ones in $\mathcal R_1$. In particular, the subtrees rooted at the children of $v_i$ that lie in $\mathcal R_3$ are processed according to the clockwise order of the edges from $v_i$ to their roots, while the role played by $\mathcal S_R$ is now assumed by $\mathcal S_L$.

The treatment of the subtrees in $\mathcal R_4$ is similar to the one of the subtrees in $\mathcal R_3$. However, when a subtree is considered, it is first horizontally translated in the interior of $\mathcal R_3$ and then processed according to the rules for such a region.

Altogether, we have described a morph $\mathcal M_{i-1,i}$ from the partially-canonical drawing $\Delta_{i-1}$ of $T[i-1]$ to $\Delta_i$, which is a partially-canonical drawing of $T[i]$ by \cref{lem:partially-canonical-delta-i}. Next, we argue about the properties of $\mathcal M_{i-1,i}$.

We first deal with the space requirements of $\mathcal M_{i-1,i}$. Consider the drawing $\Delta_0$ and place the boxes $\largebox{v}$ around the nodes $v$ of $T$; the bounding box of the arrangement of such boxes has width $w(\Delta_0)+\ell_0+4n$ and height $h(\Delta_0)+\ell_0+4n$. We claim that the drawings of $\mathcal M_{i-1,i}$ lie inside such a bounding box. Assume this is true for $\Delta_{i-1}$ (this is indeed the case when $i=1$); all subsequent drawings of $\mathcal M_{i-1,i}$ coincide with $\Delta_{i-1}$, except for the placement of the subtrees rooted at the children of $v_i$, which however lie inside $\largebox{v_i}$ in each of such drawings. Since $v_i$ has the same position in $\Delta_{i}$ as in $\Delta_0$ and since $\largebox{v_i}$ has width and height equal to $\ell_0+4n$, the claim follows.

Finally, we deal with the number of linear morphs composing $\mathcal M_{i-1,i}$. The morph $\mathcal M_{i-1,i}$ consists of the morph $\langle \Delta_{i-1},\Delta'\rangle$, followed by the morphs needed to move the subtrees rooted at the children of $v_i$ to their final positions in $\Delta_i$. Since the number of morphing steps needed to deal with each of such subtrees is constant, we conclude that $M_{i-1,i}$ consists of $O(\deg(v_i))$ linear morphing steps. This concludes the proof of \cref{lemma:partially-canonical-from-partially-canonical}.

\section{Conclusions and Open Problems}\label{se:conclusions}

We presented an algorithm that, given any two order-preserving straight-line planar grid drawings $\Gamma_0$ and $\Gamma_1$ of an $n$-node ordered tree $T$, constructs a morph $\langle \Gamma_0=\Delta_0,\Delta_1,\dots,\Delta_{k}=\Gamma_1\rangle$ such that $k$ is in $O(n)$ and such that the area of each intermediate drawing $\Delta_i$ is polynomial in $n$ and in the area of $\Gamma_0$ and $\Gamma_1$. Better bounds can be achieved if $T$ is rooted and $\Gamma_0$ and $\Gamma_1$ are also strictly-upward drawings, especially in the case in which $T$ is a binary tree.


We make two remarks about the generality of the model that we adopted. Both observations apply not only to tree drawings but, more in general, to planar graph drawings.

First, our assumption that $\Gamma_0$ and $\Gamma_1$ are grid drawings seems restrictive, and it seems more general to consider drawings that have bounded resolution, where the {\em resolution} of a drawing is the ratio between the largest and the smallest distance between a pair of geometric objects in the drawing (points representing nodes or segments representing edges). However, by using an observation from~\cite{cegl-dgppof-12}, one can argue that two morphing steps suffice to transform a drawing with resolution $r$ in a grid drawing whose area is polynomial in $r$. This is formalized in the following.

\begin{lemma} \label{le:from-resolution-to-grid}
Let $\Gamma$ be a planar straight-line drawing of a planar graph $G$ and let $r$ be the resolution of $\Gamma$. There exists a $2$-step planar morph $\langle \Gamma, \Gamma', \Gamma'' \rangle$ such that the resolution of $\Gamma'$ is $r$ and such that $\Gamma''$ is a planar straight-line grid drawing lying on an $O(r)\times O(r)$ grid.
\end{lemma}	

\begin{proof}	
First, we construct $\Gamma'$ by scaling $\Gamma$ in such a way that the smallest distance between any pair of geometric objects is $2$. Clearly, $\langle \Gamma, \Gamma' \rangle$ is a planar linear morph and the resolution of $\Gamma'$ is the same as the one of $\Gamma$; hence, the largest distance between any pair of geometric objects in $\Gamma'$ is in $O(r)$. Second, we construct $\Gamma''$ from $\Gamma'$ by moving each node to the nearest grid point; the linear morph $\langle \Gamma', \Gamma'' \rangle$ is planar since each node moves by at most $\sqrt 2/2$, hence this motion brings any two geometric objects closer by at most $\sqrt 2$, while their distance is at least $2$. Thus, $\Gamma''$ lies on an $O(r)\times O(r)$ grid.
\end{proof}

Second, our model deals with morphs that consist of sequences of drawings $\Delta_0,\Delta_1,\dots,\Delta_k$ lying on polynomial-size grids. However, no bound on the resolution is explicitly required for the drawings visualized during these morphs, i.e., for the drawings intermediate to each linear morph $\langle \Delta_i,\Delta_{i+1} \rangle$. Thus, one might wonder whether the resolution becomes arbitrarily large in some drawing of such morphs. The next lemma proves that this is not the case.

\begin{lemma}\label{le:resolution}
Let $\Gamma_0$ and $\Gamma_1$ be two planar straight-line grid drawings of the same graph $G$ lying on $\ell\times \ell$ grids, for some value $\ell>0$. Suppose that the linear morph $\langle \Gamma_0,\Gamma_1 \rangle$ is planar. Then the maximum resolution of any drawing of $\mathcal M$ is in $O(\ell^4)$.
\end{lemma}

\begin{proof}
Assume w.l.o.g.\ that the morph $\langle \Gamma_0,\Gamma_1 \rangle$ happens between times $t=0$ and $t=1$; then, for every $t\in [0,1]$, denote by $\Gamma_t$ the drawing of $\langle \Gamma_0,\Gamma_1 \rangle$ at time $t$.   	
	
Since each of $\Gamma_0$ and $\Gamma_1$ lies on a $\ell\times \ell$ grid, it follows that $\Gamma_t$ lies on a $\ell\times \ell$ grid, for every $t\in [0,1]$. Hence, the largest distance between a pair of geometric objects in $\Gamma_t$ is in $O(\ell)$. It remains to prove that the smallest distance between a pair of geometric objects in $\Gamma_t$ is in $\Omega(\frac{1}{\ell^3})$, for every $t\in [0,1]$. Such a smallest distance occurs either between two vertices of $G$ or between a vertex and an edge of $G$. Indeed, the distance between two edges always coincides with the distance between one of the edges and an end-vertex of the other edge. 

\begin{itemize}
	\item First, consider any two vertices $u$ and $v$ of $G$. We are going to prove that, for any $t\in [0,1]$, the distance between $u$ and $v$ in $\Gamma_t$ is in $\Omega(\frac{1}{\ell})$. In order to do so, it is convenient to translate $\Gamma_1$ by a vector $\vec v:= (x_{\Gamma_0}(v)-x_{\Gamma_1}(v),y_{\Gamma_0}(v)-y_{\Gamma_1}(v))$; this defines a drawing $\Gamma'_1$ of $G$ in which $v$ is placed at the same point as in $\Gamma_0$. Note that the morph $\langle \Gamma_0,\Gamma'_1 \rangle$ is planar, as for any $t\in [0,1]$ the drawing $\Gamma'_t$ at time $t$ of $\langle \Gamma_0,\Gamma'_1 \rangle$  coincides with the drawing $\Gamma_t$ translated by the vector $t \cdot \vec v$. Hence, the distance between $u$ and $v$ in $\Gamma'_t$ is the same as in $\Gamma_t$. Further, since $v$ does not move during $\langle \Gamma_0,\Gamma'_1 \rangle$, the minimum distance between $u$ and $v$ during $\langle \Gamma_0,\Gamma'_1 \rangle$ coincides with the distance between $v$ and the straight-line segment $s$ whose end-points are the positions of $u$ in $\Gamma_0$ and in $\Gamma'_1$. Since both $\Gamma_0$ and $\Gamma'_1$ lie inside a box with side-length $2\ell$ centered at $v$, it follows by Lemma~\ref{claim:minimum-grid-distance} that the distance between $v$ and $s$ is in $\Omega(\frac{1}{\ell})$.
	\item Second, consider any vertex $u$ and any edge $(v,w)$ of $G$, where $u\neq v,w$. In order to compute the minimum distance between $u$ and $(v,w)$ during $\langle \Gamma_0,\Gamma_1 \rangle$, we translate $\Gamma_0$ by a vector $\vec v:= (-x_{\Gamma_0}(v),-y_{\Gamma_0}(v))$, thus obtaining a drawing $\Gamma'_0$ of $G$ in which $v$ lies at $(0,0)$, and we translate $\Gamma_1$ by a vector $\vec v':= (-x_{\Gamma_1}(v),-y_{\Gamma_1}(v))$, thus obtaining a drawing $\Gamma'_1$ of $G$ in which $v$ also lies at $(0,0)$. Note that the morph $\langle \Gamma'_0,\Gamma'_1 \rangle$ is planar, as for any $t\in [0,1]$ the drawing $\Gamma'_t$ at time $t$ of $\langle \Gamma'_0,\Gamma'_1 \rangle$  coincides with the drawing $\Gamma_t$ translated by the vector $(1-t) \cdot \vec v + t \cdot \vec v'$. Hence, the distance between $u$ and $(v,w)$ in $\Gamma'_t$ is the same as in $\Gamma_t$. Note that $x_{\Gamma'_t}(u)=(1-t)x_{\Gamma'_0}(u)+t \cdot x_{\Gamma'_1}(u)$, and similar for $y_{\Gamma'_t}(u)$, $x_{\Gamma'_t}(w)$, and $y_{\Gamma'_t}(w)$; further, $x_{\Gamma'_t}(v)=y_{\Gamma'_t}(v)=0$. 
	
	Let $t^*\in [0,1]$ be the time of $\langle \Gamma'_0,\Gamma'_1 \rangle$ in which the distance $d$ between $u$ and $(v,w)$ is minimum; we need to prove that $d\in \Omega(\frac{1}{\ell^3})$. Let $l^{vw}_{t^*}$ be the line through $v$ and $w$ in $\Gamma'_{t^*}$ and let $d^*$ be the distance between $u$ and $l^{vw}_{t^*}$; since $(v,w)$ is part of $l^{vw}_{t^*}$ in $\Gamma'_{t^*}$, we have $d\geq d^*$. 
	
	\begin{itemize}
		\item We can assume that $d^*>0$. Indeed, if $d^*=0$, we have that $u$ lies on $l^{vw}_{t^*}$. By the planarity of $\langle \Gamma'_0,\Gamma'_1 \rangle$, we have that $u$ does not lie on the edge $(v,w)$, hence $d$ is equal to the distance between $u$ and one of the vertices $v$ and $w$, which is in $\Omega(\frac{1}{\ell})$ as proved above.
		\item We can also assume that $t^*\in (0,1)$; indeed, if $t^*=0$ or $t^*=1$, then $d,d^*\in \Omega(\frac{1}{\ell})$, by \cref{claim:minimum-grid-distance}.
	\end{itemize}
	
	The line $l^{vw}_{t^*}$ has equation $(y-y_{\Gamma'_{t^*}}(v))/(y_{\Gamma'_{t^*}}(w)-y_{\Gamma'_{t^*}}(v))=(x-x_{\Gamma'_{t^*}}(v))/(x_{\Gamma'_{t^*}}(w)-x_{\Gamma'_{t^*}}(v))$, which is $y_{\Gamma'_{t^*}}(w) \cdot x - x_{\Gamma'_{t^*}}(w) \cdot y = 0$. Then we have 
	
	\begin{equation}\label{le:distance}
	d^*=\frac{|y_{\Gamma'_{t^*}}(w) \cdot  x_{\Gamma'_{t^*}}(u)- x_{\Gamma'_{t^*}}(w) \cdot  y_{\Gamma'_{t^*}}(u)|}{\sqrt{(y_{\Gamma'_{t^*}}(w))^2+(x_{\Gamma'_{t^*}}(w))^2}}.
	\end{equation}
	
	Since $\Gamma'_{t^*}$ lies inside a box with side-length $2\ell$ centered at $v\equiv (0,0)$, it follows that $|y_{\Gamma'_{t^*}}(w)|\leq \ell$ and $|x_{\Gamma'_{t^*}}(w)|\leq \ell$, hence the denominator of \cref{le:distance} is at most $\sqrt{2\ell^2}\in O(\ell)$. It remains to prove that the numerator of \cref{le:distance} is in $\Omega(\frac{1}{\ell^2})$.
	
	The numerator of \cref{le:distance} is the absolute value of a second-degree polynomial $\mathcal P(t)$, namely \begin{eqnarray*}
	\mathcal P(t)	&:=& y_{\Gamma'_{t^*}}(w) \cdot  x_{\Gamma'_{t^*}}(u)- x_{\Gamma'_{t^*}}(w) \cdot  y_{\Gamma'_{t^*}}(u)=\\
		&=& (((1-t)y_{\Gamma'_0}(w)+t \cdot y_{\Gamma'_1}(w))\cdot ((1-t)x_{\Gamma'_0}(u)+t \cdot x_{\Gamma'_1}(u))) -\\ & & (((1-t)x_{\Gamma'_0}(w)+t \cdot x_{\Gamma'_1}(w))\cdot ((1-t)y_{\Gamma'_0}(u)+t \cdot y_{\Gamma'_1}(u)))=\\
		&=& A\cdot t^2 + B\cdot t + C, 
	\end{eqnarray*}
	
	\noindent where each of $A$, $B$, and $C$ is the algebraic sum of a constant number of terms, each of which is either an element of $\mathcal S=\{y_{\Gamma'_0}(w),y_{\Gamma'_1}(w),x_{\Gamma'_0}(w),x_{\Gamma'_1}(w),y_{\Gamma'_0}(u),y_{\Gamma'_1}(u),x_{\Gamma'_0}(u),x_{\Gamma'_1}(u)\}$, or the product of two elements in $\cal S$; note that each element in $\mathcal S$ is an integer in $[-\ell,\ell]$. Since $d^*>0$ and since the function $z=\mathcal P(t)$ is continuous, we have that $\mathcal P(t)$ is either positive over the entire interval $[0,1]$ or negative over the entire interval $[0,1]$. Assume the former, as in the latter case the absolute value in the numerator of \cref{le:distance} reverts the signs of the terms of $\mathcal P(t)$ and our arguments are the same. We then only need to prove that  $\mathcal P(t^*)\in \Omega(\frac{1}{\ell^2})$.
	
	\begin{itemize}
		\item We can assume that the minimum value of $\mathcal P(t)$ over the interval $[0,1]$ is achieved at a time $t^m\in (0,1)$. Indeed, if the minimum value of $\mathcal P(t)$ over the interval $[0,1]$ is achieved at $t=0$ or $t=1$, then we have $\mathcal P(t^*)\geq A + B + C\geq 1$ or $\mathcal P(t^*)\geq C\geq 1$ (recall that $A$, $B$, and $C$ are integers and that $\mathcal P(t)$ is positive over the entire interval $[0,1]$), respectively.
	\end{itemize}
	
	Since  $\mathcal P(t)$ has a minimum at $t^m\in (0,1)$, we have $A>0$ and the derivative $\frac{\partial \mathcal P(t)}{\partial t}=2A\cdot t+B$ has a zero in $t^m=-B/2A$. Then we have $\mathcal P(t^*)\geq \mathcal P(t^m) = A (\frac{-B}{2A})^2 + B (\frac{-B}{2A}) + C = \frac{-B^2+4AC}{4A}$. The numerator of the previous fraction is greater than or equal to $1$ (given that $A$, $B$, and $C$ are integers and that $\mathcal P(t)$ is positive over the entire interval $[0,1]$), while the denominator is in $O(\ell^2)$ (given that $A$ is the algebraic sum of a constant number of terms, each of which is either an integer in $[-\ell,\ell]$, or the product of two integers in $[-\ell,\ell]$). Hence, $\mathcal P(t^m)\in \Omega(\frac{1}{\ell^2})$ and thus $\mathcal P(t^*)\in \Omega(\frac{1}{\ell^2})$.
\end{itemize}

This concludes the proof of the lemma.
\end{proof}

\cref{le:from-resolution-to-grid,le:resolution} imply that the problem of constructing planar morphs with polynomial resolution between two planar straight-line drawings of the same planar graph can indeed be reduced to the problem of constructing planar morphs with polynomial area between two planar straight-line grid drawings of the same planar graph. For example, it follows by \cref{le:from-resolution-to-grid,le:resolution} and by \cref{th:morph-straight-line-planar-drawings} that, given an $n$-node ordered tree $T$ and given two order-preserving straight-line planar drawings $\Gamma_0$ and $\Gamma_1$ of $T$ with maximum resolution $r$, there exists an $O(n)$-step planar morph $\cal M$ from $\Gamma_0$ to $\Gamma_1$ such that the resolution of any intermediate drawing of $\cal M$ is a polynomial function of $r$.

Several problems are left open by our research. Is it possible to generalize our results to graph classes richer than trees? How about outerplanar graphs? Is it possible to improve our area bounds for morphs of straight-line planar grid drawings of trees or even just of paths? Is there a trade-off between the number of steps and the area required by a morph? Is it possible to construct upward planar morphs with a constant number of steps between any two  order-preserving strictly-upward straight-line planar grid drawings of an $n$-node rooted ordered tree? Are there other relevant tree drawing standards for which it makes sense to consider the morphing problem?

\bibliographystyle{abbrv}
\bibliography{bibliography}

\clearpage


\end{document}